\documentclass[a4paper,UKenglish,cleveref, autoref, thm-restate, standalone]{lipics-v2019}

\usepackage[utf8]{inputenc}
\usepackage[dvipsnames]{xcolor}
\usepackage{caption, setspace, subcaption}

\usepackage{amsfonts, comment}
\usepackage{mathtools,amssymb}
\usepackage{graphicx}
\usepackage{xspace}
\usepackage[T1]{fontenc}
\usepackage[mathscr]{eucal}

\usepackage{tikz,pgfplots}
\usepackage{forest}
\usetikzlibrary{fit,patterns, arrows,positioning,shadows, trees, shapes}

\graphicspath{ {./image/} }
\usepackage[linesnumbered,ruled,vlined]{algorithm2e}

\usepackage{xspace}
\usepackage{multicol}
\usepackage{amsmath}
\usepackage{caption}

\newcommand{\rsap}{{\sc Round-SAP}\xspace}
\newcommand{\rufp}{{\sc Round-UFP}\xspace}
\newcommand{\ursap}{{\sc Uni-Round-SAP}}
\newcommand{\urufp}{{\sc Uni-Round-UFPP}}
\newcommand{\rtree}{{\sc Round-Tree}}

\newcommand{\urtree}{{\sc Uni-Round-Tree}}
\newcommand{\pcolor}{{\sc Path-Coloring}}

\newcommand{\ff}{{\sc First-Fit}}

\newcommand{\eps}{\varepsilon}
\newcommand{\opt}{OPT}
\newcommand{\As}{\mathcal{A}}
\newcommand{\Ps}{\mathcal{P}}
\newcommand{\Ls}{\mathcal{L}}
\newcommand{\Rs}{\mathcal{R}}
\newcommand{\Ts}{\mathcal{T}}
\newcommand{\Ms}{\mathcal{M}}
\newcommand{\Is}{\mathcal{I}}
\newcommand{\Ws}{\mathcal{W}}
\newcommand{\Xs}{\mathcal{X}}
\newcommand{\Ys}{\mathcal{Y}}
\newcommand{\Zs}{\mathcal{Z}}

\newcommand{\Js}{\mathcal{J}}
\newcommand{\Qs}{\mathcal{Q}}
\newcommand{\Ss}{\mathcal{S}}
\newcommand{\Cs}{\mathcal{C}}

\newcommand{\Ns}{\mathcal{N}}
\newcommand{\Vs}{\mathcal{V}}
\newcommand{\Hs}{\mathcal{H}}
\newcommand{\Ds}{\mathcal{D}}
\newcommand{\sbmdm}{2-B-3-DM}

\newcommand{\js}{\mathfrak{j}}
\newcommand{\cs}{c^*}
\global\long\def\R{\mathcal{R}}%

\def\DEBUG{false} 
	\ifdefined\DEBUG

	\newcommand{\aw}[1]{\textcolor{orange}{#1}}
	\def\rem#1{{\marginpar{\raggedright\scriptsize #1}}}
	\newcommand{\arir}[1]{\rem{\textcolor{red}{$\bullet$ #1}}}
	\newcommand{\debr}[1]{\rem{\textcolor{blue}{$\bullet$ #1}}}
	\newcommand{\awr}[1]{\rem{\textcolor{orange}{$\bullet$ #1}}}
	\else

	  \newcommand{\arir}[1]{}
	  \newcommand{\debr}[1]{}
	   \newcommand{\awr}[1]{}
	\fi

\title{Approximation Algorithms for \rufp~and \rsap} 

\titlerunning{ROUND-SAP and ROUND-UFPP} 
   
\author{Debajyoti Kar} 
{Department of Computer Science and Engineering \and Indian Insitute of Technology, Kharagpur, India}{debajyoti.kar@iitkgp.ac.in}{}{}

\author{Arindam Khan}{Department of Computer Science and Automation \and Indian Institute of Science, Bengaluru, India}{arindamkhan@iisc.ac.in}{}{}

\author{Andreas Wiese}
{School of Business and Economics, Operations Analytics \and Vrije Universiteit, Amsterdam, Netherlands}{a.wiese@vu.nl}{}{}

\authorrunning{Debajyoti Kar and Arindam Khan and Andreas Wiese} 

\Copyright{Debajyoti Kar and Arindam Khan and Andreas Wiese} 



\acknowledgements{We thank Waldo Galvez, Afrouz Jabal Ameli, Siba Smarak Panigrahi and Arka Ray for helpful initial discussions.}

\nolinenumbers 

\EventEditors{John Q. Open and Joan R. Access}
\EventNoEds{2}
\EventLongTitle{42nd Conference on Very Important Topics (CVIT 2016)}
\EventShortTitle{CVIT 2016}
\EventAcronym{CVIT}
\EventYear{2016}
\EventDate{December 24--27, 2016}
\EventLocation{Little Whinging, United Kingdom}
\EventLogo{}
\SeriesVolume{42}
\ArticleNo{23}

\begin{document}

\maketitle

\begin{abstract}
We study \rufp\ and \rsap, two generalizations of the classical
\textsc{Bin Packing }problem that correspond to the unsplittable flow
problem on a path (UFP) and the storage allocation problem~(SAP), respectively.
We are given a path with capacities on its edges and a set of tasks
where for each task we are given a demand and a subpath. In \rufp,
the goal is to find a packing of all tasks into a minimum number of
copies (rounds) of the given path such that for each copy, the total
demand of tasks on any edge does not exceed the capacity of the respective
edge. In \rsap, the tasks are considered to be rectangles and the
goal is to find a non-overlapping packing of these rectangles into a
minimum number of rounds such that all rectangles lie completely below
the capacity profile of the edges.

We show that in contrast to \textsc{Bin Packing}, both the problems
do not admit an asymptotic polynomial-time approximation scheme (APTAS),
even when all edge capacities are equal. However, for this setting,
we obtain asymptotic $(2+\eps)$-approximations for both problems.
For the general case, we obtain an $O(\log\log n)$-approximation
algorithm and an $O(\log\log\frac{1}{\delta})$-approximation under
$(1+\delta)$-resource augmentation for both problems. For the intermediate
setting of the {\em no bottleneck assumption} (i.e., the maximum task demand
is at most the minimum edge capacity), we obtain absolute 
$12$- and asymptotic $(16+\eps)$-approximation algorithms for \rufp\ and
\rsap, respectively. 
\end{abstract}

\section{Introduction}

\label{sec:intro} 
The unsplittable flow on a path problem (UFP) 
and the storage allocation problem~(SAP) 
are two
well-studied problems in combinatorial optimization. In this paper,
we study \rufp~and \rsap, which are two related natural problems
that also generalize the classical {\sc Bin Packing} problem. 


In both \rufp and \rsap, we are given as input a path $G=(V,E)$
and a set of $n$ jobs $J$. We assume that $\{v_{0},v_{1},\dots,v_{m}\}$
are the vertices in $V$ from left to right and then for each $i \in \{1,...,m\}$ there is an
edge $e_{i}:=\{v_{i-1},v_{i}\}$.
Each job $\js\in J$ has integral demand $d_{\js}\in\mathbb{N}$,
a source $v_{s_{j}}\in V$, and a sink $v_{t_{j}}\in V$. We say that
each job $\js$ spans the path $P_{\js}$ which we define to be the path between
$v_{s_{\js}}$ and $v_{t_{\js}}$. For every edge $e\in E$, we are
given an integral capacity $c_{e}$. A useful geometric interpretation
of the input path and the edge capacities is the following (see Figure~\ref{closedcor2}):
consider the interval $[0,m)$ on the $x$-axis and a function $c\colon[0,m)\to\mathbb{N}$.
Each edge $e_{k}$ corresponds to the interval $[k-1,k)$ and each
vertex $v_{i}$ corresponds to the point $i$. For edge $e_{k}$, we
define $c(x)=c_{e_{k}}$ for each $x\in[k-1,k)$.

In the \rufp~problem, the objective is to partition the jobs $J$
into a minimum number of sets $J_{1},...,J_{k}$ (that we will denote
by \emph{rounds}) such that the jobs in each set $J_{i}$ form a \emph{valid
packing, }i.e., they obey the edge capacities, meaning that $\sum_{\js\in J_{i}:e\in P_{\js}}d_{\js}\le c_{e}$
for each $e\in E$. In the \rsap~problem, we require to compute
additionally for each set $J_{i}$ a non-overlapping set of rectangles
underneath the capacity profile, corresponding to the jobs in $J_{i}$
(see Figure~\ref{closedcor2}). Formally, we require for each job $\js\in J_{i}$
to determine a height $h_{\js}$ with $h_{\js}+d_{\js}\le c_{e}$
for each edge $e\in P_{\js}$, yielding a rectangle $R_{\js}=(s_{\js},t_{\js})\times(h_{\js},h_{\js}+d_{\js})$,
such that for any two jobs $\js,\js'\in J_{i}$ we have that $R_{\js}\cap R_{{\js}'}=\emptyset$.
Again, the objective is to minimize the number of rounds.

Note that unlike \rsap, in \rufp\ we do not need to pack the jobs
as contiguous rectangles. Hence, intuitively in \rufp\ we can slice the rectangles
vertically and place different slices at different heights. See Figure~\ref{closedcor2}
for the differences between the problems. 

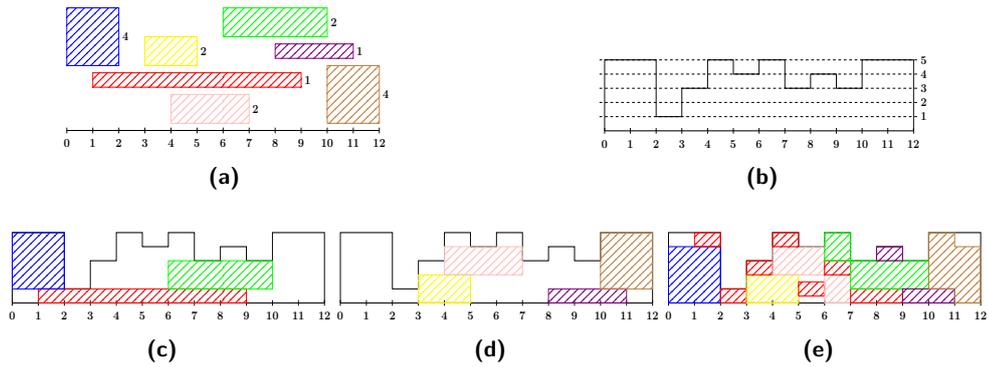
\begin{figure}
\captionsetup{justification=centering} \captionsetup[subfigure]{justification=centering}
\begin{subfigure}[b]{.5\textwidth} \centering \resizebox{4.4cm}{1.9cm}{ \begin{tikzpicture}
			\draw[thick] (1,0) -- (13,0);
			\foreach \x in {1,...,13}
				\draw[thick] (\x,0.1) -- (\x,-0.1);
			\foreach \x in {0,...,12}
				\draw[thick] (\x+1,-0.45) node {\Large \textbf{\x}};
			\draw[thick, color=pink, pattern=north east lines, pattern color = pink] (5,0.25) rectangle (8,1.25);			
			\draw[thick, color=brown, pattern=north east lines, pattern color = brown] (11,0.25) rectangle (13,2.25);
			\draw[thick, color=red, pattern=north east lines, pattern color = red] (2,1.5) rectangle (10,2);
			\draw[thick, color=blue, pattern=north east lines, pattern color = blue] (1,2.25) rectangle (3,4.25);
			\draw[thick, color=yellow, pattern=north east lines, pattern color = yellow] (4,2.25) rectangle (6,3.25);
			\draw[thick, color=violet, pattern=north east lines, pattern color = violet] (9,2.5) rectangle (12,3);
			\draw[thick, color=green, pattern=north east lines, pattern color = green] (7,3.25) rectangle (11,4.25);
			
			\draw[thick] (8.25,0.75) node {\Large \textbf{2}};
			\draw[thick] (13.25,1.25) node {\Large \textbf{4}};
			\draw[thick] (10.25,1.75) node {\Large \textbf{1}};
			\draw[thick] (3.25,3.25) node {\Large \textbf{4}};
			\draw[thick] (6.25,2.75) node {\Large \textbf{2}};
			\draw[thick] (12.25,2.75) node {\Large \textbf{1}};
			\draw[thick] (11.25,3.75) node {\Large \textbf{2}};
	
		\end{tikzpicture}} \caption{}
\label{fig:zbend1} \end{subfigure} \begin{subfigure}[b]{.5\textwidth}
\centering \resizebox{4.4cm}{1.3cm}{ \begin{tikzpicture}
			\draw[thick] (1,0) -- (13,0);
			\foreach \x in {1,...,13}
				\draw[thick] (\x,0.1) -- (\x,-0.1);
			\foreach \x in {0,...,12}
				\draw[thick] (\x+1,-0.45) node {\Large \textbf{\x}};
			\draw[thick] (1,0) -- (1,2.5) -- (3,2.5) -- (3,0.5) -- (4,0.5) -- (4,1.5) -- (5,1.5) -- (5,2.5) -- (6,2.5) -- (6,2) -- (7,2) -- (7,2.5) -- (8,2.5) -- (8,1.5) -- (9,1.5) -- (9,2) -- (10,2) -- (10,1.5) -- (11,1.5) -- (11,2.5) -- (13,2.5) -- (13,0);
			
			\foreach \x in {1,...,5}
				\draw[dashed] (0.8,\x/2) -- (13.2,\x/2);
			\foreach \x in {1,...,5}
				\draw[thick] (13.4,\x/2) node {\Large \textbf{\x}};
		\end{tikzpicture}} \caption{}
\label{fig:zbend1} \end{subfigure}

\vspace{0.5cm}
 \centering \begin{subfigure}[b]{.3\textwidth} 
 \resizebox{4.35cm}{1.2cm}{ \begin{tikzpicture}
			\draw[thick] (1,0) -- (13,0);
			\foreach \x in {1,...,13}
				\draw[thick] (\x,0.1) -- (\x,-0.1);
			\foreach \x in {0,...,12}
				\draw[thick] (\x+1,-0.45) node {\Large \textbf{\x}};
			\draw[thick] (1,0) -- (1,2.5) -- (3,2.5) -- (3,0.5) -- (4,0.5) -- (4,1.5) -- (5,1.5) -- (5,2.5) -- (6,2.5) -- (6,2) -- (7,2) -- (7,2.5) -- (8,2.5) -- (8,1.5) -- (9,1.5) -- (9,2) -- (10,2) -- (10,1.5) -- (11,1.5) -- (11,2.5) -- (13,2.5) -- (13,0);
			
		\draw[thick, color=red, pattern=north east lines, pattern color = red] (2,0) rectangle (10,0.5);
		\draw[thick, color=blue, pattern=north east lines, pattern color = blue] (1,0.5) rectangle (3,2.5);
		\draw[thick, color=green, pattern=north east lines, pattern color = green] (7,0.5) rectangle (11,1.5);	
		\end{tikzpicture}} \caption{}
\label{fig:zbend1} \end{subfigure} 
 \begin{subfigure}[b]{.3\textwidth} \centering \resizebox{4.35cm}{1.2cm}{ \begin{tikzpicture}
			\draw[thick] (1,0) -- (13,0);
			\foreach \x in {1,...,13}
				\draw[thick] (\x,0.1) -- (\x,-0.1);
			\foreach \x in {0,...,12}
				\draw[thick] (\x+1,-0.45) node {\Large \textbf{\x}};
			\draw[thick] (1,0) -- (1,2.5) -- (3,2.5) -- (3,0.5) -- (4,0.5) -- (4,1.5) -- (5,1.5) -- (5,2.5) -- (6,2.5) -- (6,2) -- (7,2) -- (7,2.5) -- (8,2.5) -- (8,1.5) -- (9,1.5) -- (9,2) -- (10,2) -- (10,1.5) -- (11,1.5) -- (11,2.5) -- (13,2.5) -- (13,0);
			
			\draw[thick, color=yellow, pattern=north east lines, pattern color = yellow] (4,0) rectangle (6,1);
			\draw[thick, color=pink, pattern=north east lines, pattern color = pink] (5,1) rectangle (8,2);
			\draw[thick, color=violet, pattern=north east lines, pattern color = violet] (9,0) rectangle (12,0.5);
			\draw[thick, color=brown, pattern=north east lines, pattern color = brown] (11,0.5) rectangle (13,2.5);
	
		\end{tikzpicture}} \caption{}
\label{fig:zbend1} \end{subfigure} 
 \begin{subfigure}[b]{.3\textwidth} \centering \resizebox{4.35cm}{1.2cm}{ \begin{tikzpicture}
			\draw[thick] (1,0) -- (13,0);
			\foreach \x in {1,...,13}
				\draw[thick] (\x,0.1) -- (\x,-0.1);
			\foreach \x in {0,...,12}
				\draw[thick] (\x+1,-0.45) node {\Large \textbf{\x}};
			\draw[thick] (1,0) -- (1,2.5) -- (3,2.5) -- (3,0.5) -- (4,0.5) -- (4,1.5) -- (5,1.5) -- (5,2.5) -- (6,2.5) -- (6,2) -- (7,2) -- (7,2.5) -- (8,2.5) -- (8,1.5) -- (9,1.5) -- (9,2) -- (10,2) -- (10,1.5) -- (11,1.5) -- (11,2.5) -- (13,2.5) -- (13,0);
			
			\draw[thick, color=blue, pattern=north east lines, pattern color = blue] (1,0) rectangle (3,2);
			\draw[thick, color=red, pattern=north east lines, pattern color = red] (2,2) rectangle (3,2.5);
			\draw[thick, color=red, pattern=north east lines, pattern color = red] (3,0) rectangle (4,0.5);
			\draw[thick, color=yellow, pattern=north east lines, pattern color = yellow] (4,0) rectangle (6,1);
			\draw[thick, color=red, pattern=north east lines, pattern color = red] (4,1) rectangle (5,1.5);
			\draw[thick, color=pink, pattern=north east lines, pattern color = pink] (5,1) rectangle (7,2);
			\draw[thick, color=red, pattern=north east lines, pattern color = red] (5,2) rectangle (6,2.5);
			\draw[thick, color=red, pattern=north east lines, pattern color = red] (6,0.25) rectangle (7,0.75);
			\draw[thick, color=pink, pattern=north east lines, pattern color = pink] (7,0) rectangle (8,1);
			\draw[thick, color=red, pattern=north east lines, pattern color = red] (7,1) rectangle (8,1.5);
			\draw[thick, color=green, pattern=north east lines, pattern color = green] (7,1.5) rectangle (8,2.5);
			\draw[thick, color=red, pattern=north east lines, pattern color = red] (8,0) rectangle (10,0.5);
			\draw[thick, color=green, pattern=north east lines, pattern color = green] (8,0.5) rectangle (11,1.5);
			\draw[thick, color=violet, pattern=north east lines, pattern color = violet] (9,1.5) rectangle (10,2);
			\draw[thick, color=violet, pattern=north east lines, pattern color = violet] (10,0) rectangle (12,0.5);
			\draw[thick, color=brown, pattern=north east lines, pattern color = brown] (11,0.5) -- (11,2.5) -- (12,2.5) -- (12,2) -- (13,2) -- (13,0) -- (12,0) -- (12,0.5) -- (11,0.5);
	
		\end{tikzpicture}} \caption{}
\label{fig:zbend1} \end{subfigure} \caption{\textbf{(a)} A set of 7 jobs with demands written beside; \textbf{(b)}
The capacity profile; \protect \\
 \textbf{(c),(d)} Any valid \rsap\ packing requires at least 2 rounds;
\protect \\
 \textbf{(e)} A valid \rufp\ packing using only 1 round.}
\label{closedcor2} 
\end{figure}

\rufp and \rsap arise naturally in the setting of resource allocation
 with connections to many fundamental optimization problems,
e.g., wavelength division multiplexing (WDM) and optical
fiber minimization (see \cite{andrews2005bounds,winkler2003wavelength}
for more details and practical motivations). The path can represent
a network with a chain of communication links in which we need to
send some required transmissions in a few rounds so that they obey the
given edge capacities. The edges can also correspond to discrete time
slots, each slot models a job that we might want to execute, and the
edge capacities model the available amount of a resource shared by
the jobs like energy or machines. \rufp\ models routing in optical
networks, where each copy of the resource corresponds to a distinct
frequency. As the number of available distinct frequencies is limited,
minimizing the number of rounds for a given set of requests is a natural
objective. \rsap\ is motivated by settings in which jobs need a
contiguous portion of an available resource, e.g., a consecutive portion
of the computer memory or a frequency bandwidth. Another application
is ad-placement, where each job is an advertisement that requires
a contiguous portion of the banner \cite{MomkeW15}.

\rufp and \rsap are APX-hard as they contain the
classical \textsc{Bin Packing} problem as a special case when $G$
has only one single edge. However, while for \textsc{Bin Packing}
there exists an asymptotic polynomial time approximation scheme (APTAS)\footnote{For formal definitions of notions like asymptotic approximation ratio, (asymptotic) polynomial time approximation scheme, etc. we refer to Appendix \ref{sec:approx}.}, it is open whether such an
algorithm exists for \rufp\ or \rsap. The best known approximation
algorithm for \rufp\ is a $O(\min\{\log n,\log m,\log\log c_{\max}\})$-approximation~\cite{JahanjouKR17}.
For the special case of \rufp~of uniform edge capacities, Pal~\cite{pal2014approximation} gave a 3-approximation. Elbassioni et al.~\cite{ElbassioniGGKN12} gave a 24-approximation algorithm for the problem under the no-bottleneck
assumption (NBA) which states that the maximum task demand is upper-bounded
by the minimum edge capacity. A result in (the full version of)~\cite{MomkeW15}
states that any solution to an instance of UFP can be partitioned
into at most 80 sets of tasks such that each of them is a solution
to the corresponding SAP instance. This immediately yields approximation algorithms
for \rsap: a 240-approximation for the case of uniform capacities, 
a 1920-approximation under the NBA, and a $O(\min\{\log n,\log m,\log\log c_{\max}\})$-approximation for the general case. These are the best known results
for~\rsap. 


\subsection{Our Contributions}

\label{subsec:ourwork} 

First, we show that both \rsap\ and \rufp, unlike the classical
\textsc{Bin Packing} problem, do not admit an APTAS, even in the uniform
capacity case. We achieve this via a gap preserving reduction from
the 3D matching problem. We create a numeric version of the problem
and define a set of hard instances for both \rsap\ and \rufp. Together
with a result of Chlebik and Chlebikova \cite{chlebik2006complexity},
we derive an explicit lower bound on the asymptotic approximation
ratio for both problems. Our hardness result holds even for the case in which
in the optimal packing no round contains more than $O(1)$ jobs, i.e., a case in which
we can even enumerate all possible packings in polynomial time.

For the case of uniform edge capacities, we give asymptotic $(2+\eps)$-approximation
algorithms for both \rufp\ and \rsap, and absolute $(2.5+\eps)$ and 3-approximation algorithms for the two problems, respectively. 
This improves upon the previous absolute
3- and 240-approximation algorithms mentioned above. Note that for
both problems our factor of 2 is a natural threshold: in many algorithms
for UFP and SAP \cite{BonsmaSW11, AnagnostopoulosGLW14, GrandoniMW017, GMWZ018, MomkeW15, MomkeW20}, 
the input tasks are partitioned 
into tasks that are relatively small and relatively
large (compared to the edge capacities). Then, both sets
are handled separately with very different sets of techniques. This
inherently loses a factor of 2. Our algorithms are based on a connection
of our problems to the dynamic storage allocation (DSA) problem and
we show how to use some known deep results for DSA~\cite{BuchsbaumKKRT03} in our
setting. In DSA the goal is to place some given tasks
as non-overlapping rectangles, minimizing the height of the resulting packing. Hence, it is somewhat surprising
that also for \rufp (where we do not have this requirement) it yields
the needed techniques for an improved approximation.

For the general cases of \rufp and \rsap\, we give an $O(\log\log \min\{m, n\})$-approximation
algorithm. Depending on the concrete values of $n,m,$ and $c_{\max}$,
this constitutes an up to exponential improvement compared to the
best known result for \rufp~\cite{JahanjouKR17} and for \rsap\ (by the reasoning via~\cite{MomkeW15} above). We divide the
input tasks into two sets: tasks that use a relatively large portion
of the capacity of \emph{at least one} of their edges (large tasks) and tasks
that use a relatively small portion of the capacity of \emph{all} of their
edges (small tasks). For each large task, we fix a corresponding rectangle
that is drawn at the maximum possible height underneath the capacity curve
(even for \rufp~for which we do not need to represent the given
tasks as rectangles) and we seek a solution in which in each round
these corresponding rectangles are non-overlapping. It follows from
known results that this loses at most a factor of~$O(1)$~\cite{BonsmaSW11, JahanjouKR17}.
Then, we solve a configuration-LP for our problem which we solve via a separation oracle. We take the integral parts of
its solution and show that the fractional parts yield several instances of our original problem 
in which each point overlaps at most $O(\log m)$ rectangles.
Using a recent result by Chalermsook and Walczak \cite{parinbar} 
on the coloring
number of rectangle intersection graphs, 
we obtain an $O(\log\log m)$-approximation.
For the small tasks of a given \rufp instance, a result in~\cite{ElbassioniGGKN12}
yields an $O(1)$-approximation, which yields our $O(\log\log m)$-approximation
for \rufp. Finally, we use the result in~\cite{MomkeW15} mentioned above in order to turn this algorithm even into a $O(\log\log m)$-approximate
solution for \rsap. Then we study the setting of resource augmentation,
i.e., where we can increase the edge capacities by a factor of $1+\delta$
while the compared optimal solution does not have this privilege.
In this case, we show that we can reduce the given problem to the
setting in which the edge capacites are in the range $[1,1/\delta)$.
Applying the algorithm from~\cite{JahanjouKR17} then yields a $O(\log\log\frac{1}{\delta})$-approximation
for this case for \rufp, and with a similar argumentation as before
also for \rsap.

Furthermore, for the case of the NBA we improve the absolute approximation
ratio from 24 to 12 for \rufp, and from 1920 to 24 for
\rsap, and we obtain even an asymptotic $(16+\eps)$-approximation for \rsap. For \rsap~we show that we can reduce the general case to the
case of uniform edge capacities, losing only a factor of~8. Thus,
future improvements for the case of uniform edge capacities will directly yield improvements 
for the case of the NBA. For \rufp~we partition the input jobs
into several sets. For some of them we reduce the given setting
to the case of unit job demands and integral edge capacities and invoke
an algorithm from \cite{NomikosPZ01} for this case. For the others, we show that 
a simple greedy routine works sufficiently well. 

If in \rufp we are given a tree instead of a path, we obtain the \rtree~problem. The best known result for it under the NBA
is a 64-approximation \cite{ElbassioniGGKN12}. We are not aware of any better result for
uniform edge capacities.
We improve the best
known approximation ratio under the NBA to 55 and also provide a 5.5-approximation algorithm for the case of uniform edge capacities.


See Table~\ref{table:results} for an overview of our results.
Due to space limitations, many proofs had to be moved to the appendix.

\begin{table}[h!]
\centering{}%
\begin{tabular}{|l|l|l|l|}
\hline 
\textbf{Problem}  & \textbf{Edge capacites} & \textbf{Previous approximation}  & \textbf{Improved approximation} \tabularnewline
\hline 
\rufp  & uniform  & 3 \cite{pal2014approximation}  & asymp. $2+\eps$, abs. 2.5$+\eps$  \tabularnewline
\rsap  & uniform  & $240$  \cite{pal2014approximation, MomkeW15} & asymp. $2+\eps$, abs. 3 \tabularnewline
\rufp  & NBA & 24  \cite{ElbassioniGGKN12} & abs. 12 \tabularnewline
\rsap  & NBA & 1920   \cite{ElbassioniGGKN12, MomkeW15} & asymp. $16+\eps$, abs. 24\tabularnewline
\rufp  & general & $O(\log\min\{n,m,\log c_{\max}\})$ \cite{JahanjouKR17}& abs. $O(\log\log \min\{n,m\})$\tabularnewline
\rsap  & general & $O(\log\min\{n,m,\log c_{\max}\})$ \cite{JahanjouKR17}& abs. $O(\log\log \min\{n,m\})$ \tabularnewline
\rufp  & general with r.a. & $O(\log\min\{n,m,\log c_{\max}\})$ \cite{JahanjouKR17}& abs. $O(\log\log(1/\delta))$\tabularnewline
\rsap  & general with r.a. & $O(\log\min\{n,m,\log c_{\max}\})$ \cite{JahanjouKR17}& abs. $O(\log\log(1/\delta))$\tabularnewline
\rtree  & uniform & 64  \cite{ElbassioniGGKN12}& asymp. 5.1, abs. $5.5$\tabularnewline
\rtree  & NBA & 64  \cite{ElbassioniGGKN12}& asymp. 49, abs. 55\tabularnewline
\hline 
\end{tabular}
\vspace{0.4cm}
\caption{Overview of our results. We distinguish the settings according to uniform edge capacities, the no-bottleneck-assumption (NBA), general edge capacities, and general edge capacities with $(1+\delta)$-resource augmentation (r.a.). Also, we distinguish
between absolute approximation ratios and asymptotic approximation ratios. All listed previous results are absolute approximation ratios.
}
\label{table:results} 
\end{table}

\subsection{Other Related Work\label{subsec:relwork}}
Adamy and Erlebach~\cite{adamy2003online}
studied \rufp~for uniform edge capacities in the online setting and gave a 195-approximation
algorithm, which was subsequently improved to 10 \cite{narayanaswamy2004dynamic,azar2006improved}.
Without the NBA, Epstein et al.~\cite{epstein2009online} showed that
no deterministic online algorithm can achieve a competitive ratio better
than $\Omega(\log\log n)$ or $\Omega(\log\log\log(c_{\max}/c_{\min}))$.
They also gave a $O(\log c_{\max}$)-competitive algorithm, where
$c_{\max}$ is the largest edge capacity. Without the NBA, recently Jahanjou
et al.~\cite{JahanjouKR17} gave 
a $O(\min(\log m,\log\log c_{\max}))$-competitive
algorithm. 

\rufp\ and \rsap\ are related to many fundamental optimization
problems. For example, \rsap~can be interpreted as an intermediate
problem between two-dimensional bin packing (2BP) and the rectangle coloring
problem (RC). In 2BP, the goal is to find an axis-parallel nonoverlapping
packing of a given set of rectangles (which we can translate in both
dimensions) into minimum number of unit square bins. If all edges
have the same capacity then \rsap~can be seen as a variant of 2BP
in which the horizontal coordinate of each item is fixed and we can
choose only the vertical coordinate. For 2BP, the present best asymptotic
approximation guarantee is 1.406 \cite{bansal2014binpacking}. On
the other hand, in RC, all rectangles are fixed and the goal is to
color the rectangles using a minimum number of colors such that no two
rectangles of the same color intersect. For RC, recently Chalermsook
and Walczak~\cite{parinbar} have given a polynomial-time algorithm that uses
only $O(\omega\log\omega)$ colors, where $\omega$ is the clique number of the corresponding
intersection graph (and hence a lower bound on the number of needed colors). Another related problem is Dynamic Storage Allocation
(DSA), where the objective is to pack the given tasks (with fixed
horizontal location) such that the maximum vertical height, $\max_{\js}(h_{\js}+d_{\js})$
(called the \textit{makespan}) is minimized. The current best known approximation ratio 
for DSA is $(2+\eps)$ \cite{BuchsbaumKKRT03}.

In a sense \rufp and \rsap~are `{\sc{Bin Packing}}-type' problems, and their corresponding
`{\sc{Knapsack}}-type' problems are UFP and SAP, respectively, where each
task has an associated profit and the goal is to select a subset of
tasks which can be packed into one single round satisfying the corresponding
valid packing constraints. There is a series of work \cite{GrandoniMW017,bansal2014logarithmic,BatraG0MW15,AnagnostopoulosGLW14,GMWZ018, GMW22a}
in UFP, culminating in a PTAS~\cite{GMW22b}. It is maybe surprising that
\rufp~does not admit an APTAS, even though  UFP admits a PTAS. 
For SAP, the currently best polynomial time approximation ratio is
$2+\eps$ \cite{MomkeW15}, which has been recently improved to $1.969+\eps$
\cite{MomkeW20} for the case of uniform capacities, and also a quasi-polynomial time
$(1.997+\eps)$-approximation is known for quasi-polynomially bounded input data.


There are many other related problems, such as two-dimensional knapsack
\cite{GalvezGHI0W17,Jansen2004, Galvez00RW21}, strip packing \cite{GGIK16,Galvez0AJ0R20, DeppertJ0RT21},
maximum independent set of rectangles \cite{AHW19,parinbar,mitchellFocs, Galvez22}, guillotine
separability of rectangles \cite{KMR20, khan2021guillotine, KS21}, weighted bipartite edge
coloring \cite{KhanS15}, maximum edge disjoint paths \cite{chuzhoy2016polylogarithmic},
etc. We refer the readers to \cite{khan2015approximation,CKPT17}
for an overview of these problems.

\section{Preliminaries}

\label{sec:prelim} 
Let $\opt_{UFP}$ and $\opt_{SAP}$ denote the optimal number of rounds required to
pack all jobs of a given instance of \rufp and \rsap, respectively. By simple preprocessing, we can assume that each vertex
in $V$ corresponds to endpoint(s) of some job(s) in $J$, and hence $m\le2n-1$.
For each job $\js$ denote by $\mathfrak{b}_{\js}$ the minimum capacity of 
the edges in $P_{\js}$, i.e., $\min\{c_{e}:e\in P_{\js}\}$ which we denote as the \textit{bottleneck
capacity} of $\js$. Job $\js$
is said to \textit{pass through} edge $e$ if $e\in P_{\js}$. The \textit{load} on edge $e$ is defined as $l_{e}:=\sum_{e\in P_{\js}}d_{\js}$,
the total sum of demands of all jobs passing through $e$. Let $L:=\max_{e}l_{e}$
denote the maximum \textit{load}. The \textit{congestion} $r_{e}$
of edge $e$ is defined as $r_{e}:=\lceil l_{e}/c_{e}\rceil$, and
$r:=\max_{e}r_{e}$ denotes the maximum \textit{congestion}.
Clearly, $r$ is a lower bound on $\opt_{UFP}$ and $\opt_{SAP}$. 




\section{Lower Bounds}

\label{section:lowerboundnew} A simple reduction from the \textsc{Partition}
problem shows that it is NP-hard to obtain a better approximation ratio 
than $3/2$
for the classical \textsc{Bin Packing} problem. However, in the resulting
instances, the optimal solutions use only two or three bins. On the
other hand, \textsc{Bin Packing} admits an APTAS~\cite{de1981bin} and thus, 
for any $\eps>0$, a $(1+\eps)$-approximation algorithm for
instances in which $OPT$ is sufficiently large. 
Since \rsap\ and \rufp\ are generalizations of \textsc{Bin Packing} (even if $G$ has only a single edge),
the lower bound of $3/2$ continues to hold. However, maybe surprisingly,
we show that unlike \textsc{Bin Packing}, \rsap\ and \rufp\ do
not admit APTASes, even in the case of  uniform edge capacities. More
precisely, we provide a lower bound of $(1+1/1398)$ on the asymptotic
approximation ratio for \rsap\ and \rufp\ via a reduction from
the \textsc{2-Bounded Occurrence Maximum 3-Dimensional Matching} (\sbmdm)
problem. 

In \sbmdm, we are given as input three pairwise disjoint sets $X:=\{x_{1},x_{2},\dots,x_{q}\}$,
$Y:=\{y_{1},y_{2},\dots,y_{q}\}$, and $Z:=\{z_{1},z_{2},\dots,z_{q}\}$
and a set of triplets $\Ts\subseteq X\times Y\times Z$ such that
each element of $X\cup Y\cup Z$ occurs in exactly two triplets in
$\Ts$. Note that $|X|=|Y|=|Z|=q$ and $|\Ts|=2q$. A {\em matching}
is a subset $M\subseteq\Ts$ such that no two triplets in $M$ agree
in any coordinate. The goal is to find a matching of maximum cardinality
(denoted by $OPT_{3DM}$). Chlebik and Chlebikova \cite{chlebik2006complexity}
gave the following hardness result. 

\begin{theorem} [\cite{chlebik2006complexity}]\label{theorem:chlebikhardness_new}
For \sbmdm~there exists a family of instances such that for each
instance $K$ of the family, either $OPT_{3DM}(K)<\alpha(q):=\lfloor0.9690082645q\rfloor$
or $OPT_{3DM}(K)\ge\beta(q):=\lceil0.979338843q\rceil$, and it is
NP-hard to distinguish these two cases. 

\end{theorem}

Hardness of \sbmdm~has been useful in inapproximability results
for various (multidimensional) packing, covering, and scheduling problems,
e.g. vector packing \cite{woeginger1997there}, geometric bin packing
\cite{bansal2006bin}, geometric bin covering \cite{chlebik2009hardness},
generalized assignment problem \cite{chekuri2005polynomial}, etc.
Similar to these results, we also use gadgets based on a reduction
from \sbmdm\ to the 4-Partition problem. 
However, the previous techniques are not directly transferable to
our problem due to the inherent differences between these problems.
Therefore, we first use the technique from \cite{woeginger1997there}
to associate certain integers with the elements of $X\cup Y\cup Z$
and $\Ts$ and then adapt the numeric data in a different way to obtain
the hard instances. 

Let $\rho=32q$ and let $\Vs$ be the set of $5q$ integers defined
as follows: $x'_{i}=i\rho+1\text{, for }1\le i\le|X|$, $y'_{j}=j\rho^{2}+2\text{, for }1\le j\le|Y|$,
$z'_{k}=k\rho^{3}+4\text{, for }1\le k\le|Z|,$ $\tau'_{l}=\rho^{4}-k\rho^{3}-j\rho^{2}-i\rho+8\text{, for each triplet }\tau_{l}=(x_{i},y_{j},z_{k})\in\Ts.$
Define $\gamma=\rho^{4}+15$. The following result is due to Woeginger
\cite{woeginger1997there}.

\begin{lemma} [\cite{woeginger1997there}]\label{lemma:woegingerlemma}
Four integers in $\Vs$ sum up to the value $\gamma$ if and only
if (i) one of them corresponds to some element $x_{i}\in X$, one
to some element $y_{j}\in Y$, one to an element $z_{k}\in Z$, and
one to some triplet $\tau_{l}\in\Ts$, and if (ii) $\tau_{l}=(x_{i},y_{j},z_{k})$
holds for these four elements. \end{lemma}

Now, we create a hard instance, tailor-made for our problems. We define
that our path $G=(V,E)$ has $40000\gamma$ vertices
that we identify with the numbers $0,1,...,40000\gamma$. For each
$x_{i}\in X$ (respectively $y_{j}\in Y$, $z_{k}\in Z$), we specify
two jobs $a_{X,i}$ and $a'_{X,i}$ (respectively $a_{Y,j}$, $a'_{Y,j}$
and $a_{Z,k}$, $a'_{Z,k}$), which will be called \textit{peers}
of each other. Each job $\js$ is specified by a triplet $(s_{\js},t_{\js},d_{\js})$.
We define
\begin{itemize}
\item $a_{X,i}=(0,20000\gamma-4x'_{i},999\gamma+4x'_{i})\text{ and }a'_{X,i}=(20000\gamma-4x'_{i},40000\gamma,1001\gamma-4x'_{i})$,
\item $a_{Y,j}=(0,20000\gamma-4y'_{j},999\gamma+4y'_{j})\text{ and }a'_{Y,j}=(20000\gamma-4y'_{j},40000\gamma,1001\gamma-4y'_{j})$,
and 
\item $a_{Z,k}=(0,20000\gamma-4z'_{k},999\gamma+4z'_{k})\text{ and }a'_{Z,k}=(20000\gamma-4z'_{k},40000\gamma,1001\gamma-4z'_{k})$. 
\end{itemize}
For each $\tau_{l}\in\Ts$, we define two jobs $b_{l}$ and $b'_{l}$
(also \textit{peers}) by: 
\begin{itemize}
\item $b_{l}=(0,19001\gamma-4\tau'_{l},999\gamma+4\tau'_{l})\text{ and }b'_{l}=(19001\gamma-4\tau'_{l},40000\gamma,1001\gamma-4\tau'_{l})$. 
\end{itemize}
Finally let $D$ be a set of $5q-4\beta(q)$ \emph{dummy} jobs each
specified by $(0,40000\gamma,2997\gamma)$. We define that each edge
$e\in E$ has a capacity of $c_{e}:=\cs:=4000\gamma$. This completes
the reduction. For any job $\js=(s_{\js},t_{\js},d_{\js})$ we define
its \emph{width} $w_{\js}:=t_{\js}-s_{\js}$. 

Let $A_{X}:=\{a_{X,i}\mid1\le i\le q\}$ and $A'_{X}:=\{a'_{X,i}\mid1\le i\le q\}$.
The sets $A_{Y}$, $A'_{Y}$, $A_{Z}$, $A'_{Z}$ are defined analogously.
Let $A:=A_{X}\cup A_{Y}\cup A_{Z}$ and $A':=A'_{X}\cup A'_{Y}\cup A'_{Z}$.
Finally let $B:=\{b_{l}\mid1\le l\le2q\}$ and $B':=\{b'_{l}\mid1\le l\le2q\}$. 

To provide some intuition, we first give an upper bound on the number
of jobs that can be packed in a round. All following lemmas, statements,
and constructions hold for both \rsap\ and \rufp. 
\begin{lemma}
\label{lem:maxnumberofjobs}
In any feasible solution any round can contain at most 8 jobs. 
\end{lemma}

We say that a round is \emph{nice }if it contains exactly 8 jobs.
It turns out that such a round corresponds exactly to one element
$\tau_{l}=(x_{i},y_{j},z_{k})\in\Ts$. We say that the jobs $a_{X,i},a'_{X,i},a_{Y,j},a'_{Y,j},a_{Z,k},a'_{Z,k},b_{l},$
and $b'_{l}$ \emph{correspond }to $\tau_{l}=(x_{i},y_{j},z_{k})$. 
\begin{lemma}
\label{lem:propertyofnicerounds}
We have that a round is nice if and only if there is an element $\tau_{l}=(x_{i},y_{j},z_{k})\in\Ts$
such that the round contains exactly the jobs that correspond to $\tau_{l}=(x_{i},y_{j},z_{k})$. 
\end{lemma}

Given an optimal solution $OPT_{3DM}$ to \sbmdm~ with $|OPT_{3DM}|\ge\beta(q)$, 
we construct a solution as follows:
\begin{enumerate}
\item Let $\Ms$ be any subset of $OPT_{3DM}$ with $|\Ms|=\beta(q)$. Create $\beta(q)$ \textit{\emph{nice}}\textit{ }rounds corresponding
to the elements in $\Ms$, i.e., for each element $\tau_{l}=(x_{i},y_{j},z_{k})\in \Ms$,
create a round containing the jobs that correspond to $\tau_{l}=(x_{i},y_{j},z_{k})$. 
\item For each $\tau_{l}\in\Ts\setminus \Ms$ create a round containing $b_l$ and $b'_l$ along with a dummy job. 
\item For each $x_{i}\in X$ (respectively $y_{j}\in Y$, $z_{k}\in Z$)
not covered by $\Ms$, pack $a_{X,i}$ and $a'_{X,i}$ (respectively
$a_{Y,j}$, $a'_{Y,j}$ and $a_{Z,k}$, $a'_{Z,k}$) together with
one dummy job in one round. 
\end{enumerate}
\begin{lemma}
\label{lem:reduction-1}If $|OPT_{3DM}|\ge\beta(q)$ then the constructed
solution is feasible and it uses at most $5q-3\beta(q)$ rounds.
\end{lemma}

\begin{proof}
One can easily check that all constructed rounds are feasible. In
step (1) we construct exactly $\beta(q)$ rounds. In step (2), we
construct $|\Ts|-\beta(q) = 2q-\beta(q)$ rounds, since
$|\Ts|=2q$. In step (3), we construct $3|\Ts\setminus \Ms|=3q-3|OPT_{3DM}|$
rounds. Hence, overall we construct at most $5q-3|OPT_{3DM}|\le5q-3\beta(q)$
rounds.
\end{proof}
Conversely, assume that $|OPT_{3DM}| < \alpha(q)$ and that we are
given any feasible solution to our constructed instance. We want to
show that it uses at least $5q-3\beta(q)+\frac{1}{7}(\beta(q)-\alpha(q))$
rounds. For this, a key property of our construction is given in the
following lemma.
\begin{lemma}
\label{lem:propertyofdummyrounds}
If a round contains a dummy job, then it can have at most three jobs:
at most one dummy job, at most one job from $A\cup B$, and at most one job from $A'\cup B'$. 
\end{lemma}

Let $n_{g}$ denote the number of nice rounds in our solution, $n_{d}$
the number of rounds with a dummy job, and $n_{b}$ the number of
remaining rounds. Note that each of the latter rounds can contain
at most 7 jobs each. Since all jobs in $A\cup A' \cup B \cup B'$ need to be assigned to a round,
we have that $8n_{g}+7n_{b}+2n_{d}\ge6q+2|\Ts|=10q$. Since the nice rounds correspond to a matching of the given instance
of~\sbmdm, we have that $n_{g}\le\alpha(q)$. Using this, we lower-bound
the number of used rounds in the following lemma.
\begin{lemma}
\label{lem:reduction-2}If $|OPT_{3DM}|< \alpha(q)$ then the number
of rounds in our solution is $n_{d}+n_{g}+n_{b}\ge(5q-3\beta(q))+\frac{1}{7}(\beta(q)-\alpha(q))$.
\end{lemma}

\begin{proof}
Since $8n_{g}+7n_{b}+2n_{d}\ge6q+2|\Ts|=10q$ and $n_{d}=5q-4\beta(q)$,
we obtain that $8n_{g}+7n_{b}\ge8\beta(q)$. Thus $n_{g}+n_{b}\ge\frac{8}{7}\beta(q)-\frac{1}{7}n_{g}$.
Since $n_{g}\le\alpha(q)$ the number of rounds is at least $n_{d}+n_{g}+n_{b}\ge5q-4\beta(q)+\frac{8}{7}\beta(q)-\frac{1}{7}\alpha(q)=(5q-3\beta(q))+\frac{1}{7}(\beta(q)-\alpha(q))$.
\end{proof}
Now Lemmas~\ref{lem:reduction-1} and \ref{lem:reduction-2} yield
our main theorem.

\begin{theorem} \label{theorem:hardnessproofnew-1} There exists
a constant $\delta_{0}>1/1398$, such that it is NP-hard to approximate
\rufp~and \rsap~in the case of uniform edge capacities with an
asymptotic approximation ratio less than $1+\delta_{0}$. \end{theorem} 

\section{Algorithms for Uniform Capacity Case}

\label{sec:Uniform_capacity} In this section, we provide asymptotic
$(2+\eps)$-approximation for \rsap\ and \rufp\ for the case of
uniform edge capacities. %
We distinguish two cases, depending on the value of $d_{\max} := \max_{\js \in J} d_{\js}$ compared to $L$.

\subsection{Case 1: $d_{\max}\le\eps^{7}L$.\label{subsubsec:easysubsec}}

First, we invoke an algorithm from \cite{BuchsbaumKKRT03} for the
dynamic storage allocation (DSA) problem. Recall that in DSA 
the input consists of a set of jobs like in \rsap\ and \rufp, but
without upper bounds of the edge capacities. Instead, we seek
to define a height $h_{\js}$ for each job $\js$ such that the resulting
rectangles for the jobs are non-overlapping and the \emph{makespan}
$\max_{\js} (h_{\js}+d_{\js})$ is minimized. The maximum load $L$ is
defined as in our setting. 

We invoke the following theorem on our input jobs $J$ with $\delta:=\eps$.

\begin{theorem} [\cite{BuchsbaumKKRT03}]\label{theorem:thorrup}
Assume that we are given a set of jobs $J'$ such that $d_{\js}\le\delta^{7}L$
for each job $\js\in J'$. Then there exists an algorithm that produces
a DSA packing of $J'$ with makespan at most $(1+\kappa \delta)L$, where $\kappa >0$ is some global constant independent of $\delta$. \end{theorem}

Let $\xi$ denote the makespan of the resulting solution to DSA and
let $\cs$ denote the (uniform) edge capacity. 
For each $h\in \mathbb{R}$, we define the horizontal line $\ell_h:=\mathbb{R}\times\{h\}$. 
A job $\js$ is said to be \textit{sliced} by $\ell_{h}$ if for the computed packing
of the jobs it holds that
$h_{\js}<h<h_{\js}+d_{\js}$.
Now we will transform this into \rsap~or \rufp~packing. 

We define a set of
rounds $\Gamma_{1}$. The set $\Gamma_{1}$ contains a round for each integer $i$
with $0\le i\le\lfloor\xi/\cs\rfloor$ and this round contains all
jobs lying between $\ell_{i\cs}$ and $\ell_{(i+1)\cs}$. Thus, 
$|\Gamma_{1}|\le\lfloor(1+\kappa\eps)L/\cs)\rfloor+1\le(1+\kappa\eps)r+1$.
There are two subcases.

{\em Subcase A: Assume that $r > {1}/{(2\kappa \eps)}$}. In this case $|\Gamma_1| \le (1+3\kappa\eps)r$. We define a set of rounds $\Gamma_{2}$ as follows. 
For each
integer $i$ with $1\le i\le\lfloor\xi/\cs\rfloor$,  $\Gamma_{2}$ has a round containing
all jobs that are sliced by $\ell_{i\cs}$. Thus $|\Gamma_2|\le \lfloor (1+\kappa \eps)L/c^* \rfloor \le (1+\kappa\eps)r$. Hence, the total number of rounds is bounded by $(2+4\kappa\eps)r \le (2+O(\eps))\opt_{UFP}\le (2+O(\eps))\opt_{SAP}$.

{\em Subcase B: Assume that $r\le {1}/{(2\kappa \eps)}$}. Now $\xi \le (1+\kappa\eps)L$, and therefore $\xi -L \le \kappa\eps L \le c^*/2$. Hence, we have $|\Gamma_1|\le r+1$ and the $(r+1)^{\text{th}}$ round is filled up to a capacity of at most $c^*/2$ on each edge. Now the total load of the set of jobs that are sliced by $(\ell_{ic^*})_{1\le i\le r}$ is at most $r\cdot \eps^7L$. We now invoke the following result on DSA to this set of jobs.

\begin{theorem}[\cite{Gergov99}]\label{thm:gergov'sdsaresult}
Let $J'$ be a set of jobs with load $L$. Then a DSA packing of $J'$ of makespan at most $3L$ can be computed in polynomial time.
\end{theorem}

Thus the makespan of the computed solution is at most $3r\cdot \eps^7L \le 3\cdot \frac{1}{2\kappa\eps}\cdot \eps^7\cdot \frac{c^*}{2\kappa\eps}\le c^*/2$, if $\eps$ is small enough. Hence these jobs can added to the $(r+1)^{\text{th}}$ round of $\Gamma_1$. Therefore, we get a packing of $J$ using at most $r+1\le \opt_{UFP}+1\le \opt_{SAP}+1$ rounds.

\subsection{Case 2: $d_{\max}>\eps^{7}L$\label{subsubsec:hardsubsubsec} }

For this case, we have $c^* \ge d_{\max}> \eps^7 L$ and therefore $r = \lceil L/c^* \rceil \le 1/\eps^7$. We partition the input jobs into \emph{large }and \emph{small} jobs
by defining $J_{\mathrm{large}}:=\{\js\in J|d_{\js}>\eps^{56}L\}$
and $J_{\mathrm{small}}:=\{\js\in J|d_{\js}\le\eps^{56}L\}$. 


We start with the small jobs $J_{\mathrm{small}}$. First, we apply
Theorem~\ref{theorem:thorrup} to them with $\delta:=\eps^{56}$ and
obtain a DSA packing $\Ps$ for them. We transform it into a solution
to \rsap\ with at most $r+1$ rounds as follows: we introduce a set $\Gamma_1$ consisting of $r+1$ rounds exactly as in the previous case (when $r\le {1}/{(2\kappa \eps)}$). The $(r+1)^{\text{th}}$ round would be filled up to a capacity of at most $\kappa\eps^8 L \le \kappa\eps c^*$. Again applying \Cref{thm:gergov'sdsaresult} to the remaining jobs, we get a DSA packing of makespan at most $3r \cdot \eps^{56}L \le 3\cdot ({1}/{\eps^7})\cdot \eps^{56}\cdot ({c^*}/{\eps^7})\le c^*/2$, and therefore these jobs can be packed inside the $(r+1)^{\mathrm{th}}$ round. Hence, there exists a packing of $J_{\text{small}}$ using at most $r+1\le \opt_{UFP}+1\le \opt_{SAP}+1$ rounds.

Now we consider the large jobs $J_{\mathrm{large}}$. Our strategy
is to compute an optimal solution for them via dynamic programming
(DP). Intuitively, our DP orders the jobs in $J_{\mathrm{large}}$
non-decreasingly by their respective source vertices and assigns them
to the rounds in this order. Since the jobs are large, each edge is
used by at most $1/\eps^{56}$ large jobs, and using interval coloring one can
show easily that at most $1/\eps^{56}=O_{\eps}(1)$ rounds
suffice (e.g., we can color the jobs with $1/\eps^{56}$ colors such that 
no two jobs with intersecting paths have the same color). 
In our DP we have a cell for each combination of an edge $e$ and the assignment
of all jobs passing through $e$ to the rounds. Given this, the corresponding subproblem is to assign additionally 
all jobs to the rounds whose paths lie completely on the right of $e$.

For \rsap\ we additionally want to bound the number of
possible heights $h_{\js}$. To this end, we restrict ourselves to
packings that are normalized which intuitively means that all jobs
are pushed up as much as possible. Formally, we say that a packing
for a set of jobs $J'$ inside a round is \textit{normalized} if for every $\js\in J'$,
either $h_{\js}+d_{\js}=\cs$ or $h_{\js}+d_{\js}=h_{\js'}$ for some
$\js'\in J'$ such that $P_{\js}\cap P_{\js'}\neq\emptyset$ (see \Cref{fig:chain} in \Cref{sec:normalizedpacking}). 

\begin{lemma}
\label{lem:normal}
Consider a valid packing of a set of jobs $J'\subseteq J_{\mathrm{large}}$
inside one round. Then there is also a packing for $J'$ that is normalized.
\end{lemma}

Now the important insight is that in a normalized packing of large
jobs, the height $h_{\js}$ of a job $\js$ is the difference of (the
top height level) $\cs$ and the sum of at most $1/\eps^{56}$
jobs in $J_{\mathrm{large}}$. Thus, the number of possible heights
is bounded by $n^{O(1/\eps^{56})}$ and we can compute all these
possible heights before starting our DP. 
\begin{lemma}
\label{lemma:heights}Given $J_{\mathrm{large}}$ we can compute a set $\mathcal{H}$ of $n^{O(1/\eps^{56})}$
values such that in any normalized packing of a set $J'\subseteq J_{\mathrm{large}}$
inside one round, the height $h_{\js}$ of each job $\js\in J'$ is
contained in $\mathcal{H}$.
\end{lemma}

Now we can compute the optimal packing via a dynamic program as described
above, which yields the following lemma. 

\begin{lemma} \label{lemma:maindynamicprogramming} Consider an instance
of \rufp\ or \rsap\ with a set of jobs $J'$ satisfying the following
conditions: 
\begin{enumerate}[(i)] 
\item The number of jobs using any edge is bounded by $\omega$. 
\item In the case of \rsap\ there is a given set $\mathcal{H}'$ of allowed
heights for the jobs. 
\end{enumerate}
Then we can compute an optimal solution to the given instance in time $(n|\mathcal{H}'|)^{O(\omega)}$. \end{lemma}

We invoke Lemma~\ref{lemma:maindynamicprogramming} with $J':=J_{\mathrm{large}}$,
$\omega=1/\eps^{56}$, and in the case of \rsap\ we define $\mathcal{H}'$
to be the set $\mathcal{H}$ due to Lemma~\ref{lemma:heights}. This
yields at most $OPT$ rounds in total for the large jobs $J_{\mathrm{large}}$. Hence,  we obtain a packing of $J$ using at most $2\cdot \opt +1$ rounds.

Case 1 and 2 together yields our main theorem for the case of uniform edge capacities.

\begin{theorem} \label{theorem:twoasymptoticapprox} For any $\eps>0$,
there exist asymptotic $(2+\eps)$-approximation algorithms for \rsap\ and
\rufp, assuming uniform edge capacities. \end{theorem}

We now derive some bounds on the absolute approximation ratios. If $\opt_{SAP}= 1$, our algorithm would return a packing using at most $(2+\eps)\cdot 1 + 1$ rounds, and hence at most 3 rounds. If $\opt_{SAP}\ge 2$, then our algorithm uses at most $(2+\eps)\cdot OPT_{SAP}+ OPT_{SAP}/2 = (2.5+\eps)OPT_{SAP}$ rounds. Hence, we obain the following result.

\begin{theorem}
\label{thm:unisapabs}
There exists a polynomial time $3$-approximation algorithm for \rsap, assuming uniform edge capacities.
\end{theorem}

For \rufp, it is easy to check whether $\opt_{UFP}=1$ by checking whether $L \le c^*$. Otherwise, $\opt_{UFP}\ge 2$ and similar as above, the number of rounds used would be at most $(2.5+\eps)\cdot OPT_{UFP}$. This gives an improvement over the result of Pal \cite{pal2014approximation}.

\begin{theorem}
For any $\eps >0$, there exists a polynomial time $(2.5+\eps)$-approximation algorithm for \rufp, assuming uniform edge capacities.
\end{theorem}

\section{General Case}
In this section, present our algorithms for the general cases of \rufp and \rsap. We begin with 
our $O(\log\log \min\{n,m\})$-approximation algorithms where we 
consider \rufp\ first and describe later how to extend our algorithm
to \rsap. We split the input jobs into large and small
jobs. We define $J_{\mathrm{large}}:=\{\js\in J|d_{\js}>\mathfrak{b}_{\js}/4\}$
and $J_{\mathrm{small}}:=\{\js\in J|d_{\js}\le \mathfrak{b}_{\js}/4\}$.
For the small jobs, we invoke a result by Elbassioni et al. \cite{ElbassioniGGKN12}
that yields a 16-approximation.
\begin{theorem}[\cite{ElbassioniGGKN12}]
\label{thm:Elbassioni-small} We are given an instance of \rufp\ with
a set of jobs $J'$ such that $d_{\js}\le\frac{1}{4}\mathfrak{b}_{\js}$ for
each job $\js\in J'$. Then there is a polynomial time algorithm that computes
a 16-approximate solution to $J'$. 
\end{theorem}

Now consider the large jobs $J_{\mathrm{large}}$. For each job $\js\in J_{\mathrm{large}}$, 
we define a rectangle $R_{\js}=(s_{\js},\mathfrak{b}_{\js}-d_{\js})\times(t_{\js},\mathfrak{b}_{\js})$.
Note that $R_{\js}$ corresponds to the rectangle for $\js$ in \rsap
if we assign $\js$ the maximum possible height $h_{\js}$ (which is $h_{\js}:=\mathfrak{b}_{\js}-d_{\js}$). We say
that a set of jobs $J'\subseteq J_{\mathrm{large}}$ is \emph{top-drawn}
(underneath the capacity profile), if their rectangles are pairwise
non-overlapping, i.e., if $R_{\js}\cap R_{\js'}=\emptyset$ for any
$\js,\js'\in R_{\js}$. If a set of jobs $J'$ is top-drawn, then
it clearly forms a feasible round in \rufp. 
However, not every feasible round of \rufp is top-drawn. Nevertheless, we look for a solution to \rufp in which the
jobs in each round are top-drawn. The following lemma implies that
this costs only a factor of 8 in our approximation ratio.

\begin{lemma}[\cite{BonsmaSW11}] \label{lemma:itspacking}  Let $J' \subseteq J_{\mathrm{large}}$ 
be a set of jobs packed in  a feasible round for a given instance of \rufp. Then $J'$ can
be partitioned into at most 8 sets such that each of them is top-drawn.\end{lemma}

Let $\R_{\mathrm{large}}:=\{R_{\js}|\js\in J_{\mathrm{large}}\}$
denote the rectangles corresponding to the large jobs and let $\omega_{\mathrm{large}}$ be 
their clique number, i.e., the size of the largest set $\R'\subseteq\R_{\mathrm{large}}$
such that all rectangles in $\R'$ pairwise overlap. As a consequence of Helly's theorem (see \cite{helly22} for details), 
 for such a set of axis-parallel rectangles $\R'$ there must be a point in which all rectangles
in $\R'$ overlap. Note that we need at least $\omega_{\mathrm{large}}$
rounds since we seek a solution with only top-drawn jobs in
each round. 

We first reduce the original instance to the case where there are only $O(m^2)$ many distinct job demands. For each $i \in \{1,\ldots,m\}$, let $p_{i}$ denote the point $(i,c_{e_i})$. We draw a horizontal and a vertical line segment passing through $p_i$ and lying completely under the capacity profile (see \Cref{processingtopdrawn}(a)). This divides the region underneath the capacity profile into at most $m^2$ regions. Let $\Hs$ denote the set of horizontal lines and $\Vs$ denote the set of vertical lines drawn. Thus, the top edge of any rectangle corresponding to a large job must touch a line in $\Hs$. Now consider any rectangle $R_{\js}$ corresponding to a job $\js$. Let $h\in \Hs$ be the horizontal line segment lying just below the bottom edge of $R_{\js}$. We increase the value of $d_{\js}$ so that the the bottom edge of $R_{\js}$ now touches the line segment $h$ (see \Cref{processingtopdrawn}(b)). Since the rectangles were top-drawn, the clique number of this new set of large rectangles (denoted by $\R'_{\mathrm{large}}$) does not change. Also any feasible packing of $\R'_{\mathrm{large}}$ is a feasible packing of $\R_{\mathrm{large}}$.

\begin{figure}
\captionsetup{justification=centering} \captionsetup[subfigure]{justification=centering}
\begin{subfigure}[b]{.5\textwidth}
\centering \resizebox{4.4cm}{1.3cm}{ \begin{tikzpicture}
			\draw[thick] (1,0) -- (13,0);
			\foreach \x in {1,...,13}
				\draw[thick] (\x,0.1) -- (\x,-0.1);
			\foreach \x in {0,...,12}
				\draw[thick] (\x+1,-0.45) node {\Large \textbf{\x}};
			\draw[thick] (1,0) -- (1,2.5) -- (3,2.5) -- (3,0.5) -- (4,0.5) -- (4,1.5) -- (5,1.5) -- (5,2.5) -- (6,2.5) -- (6,2) -- (7,2) -- (7,2.5) -- (8,2.5) -- (8,1.5) -- (9,1.5) -- (9,2) -- (10,2) -- (10,1.5) -- (11,1.5) -- (11,2.5) -- (13,2.5) -- (13,0);
			
			\draw[dashed] (1,0.5) -- (13,0.5);
			\draw[dashed] (4,1.5) -- (13,1.5); 
			\draw[dashed] (5,2) -- (8,2);
			\draw[dashed] (3,0) -- (3,0.5);			
			\draw[dashed] (4,0) -- (4,0.5);
			\draw[dashed] (5,0) -- (5,1.5);
			\draw[dashed] (6,0) -- (6,2);
			\draw[dashed] (7,0) -- (7,2);
			\draw[dashed] (8,0) -- (8,1.5);
			\draw[dashed] (9,0) -- (9,2);
			\draw[dashed] (10,0) -- (10,2);
			\draw[dashed] (11,0) -- (11,2);
			
			\draw[thick, color=blue, pattern=north east lines, pattern color = blue] (1,1.5) rectangle (3,2.5);
			\draw[thick, color=red, pattern=north east lines, pattern color = red] (7,1) rectangle (11,1.5);
			\draw[thick, color=green, pattern=north east lines, pattern color = green] (11,1) rectangle (13,2.5);
			\draw[thick, color=yellow, pattern=north east lines, pattern color = yellow] (5,1.5) rectangle (6,2.5);
			\draw[thick, color=violet, pattern=north east lines, pattern color = violet] (4,1) rectangle (6,1.5);
		\end{tikzpicture}} \caption{}
\label{fig:zbend1} \end{subfigure}
\begin{subfigure}[b]{.5\textwidth}
\centering \resizebox{4.4cm}{1.3cm}{ \begin{tikzpicture}
			\draw[thick] (1,0) -- (13,0);
			\foreach \x in {1,...,13}
				\draw[thick] (\x,0.1) -- (\x,-0.1);
			\foreach \x in {0,...,12}
				\draw[thick] (\x+1,-0.45) node {\Large \textbf{\x}};
			\draw[thick] (1,0) -- (1,2.5) -- (3,2.5) -- (3,0.5) -- (4,0.5) -- (4,1.5) -- (5,1.5) -- (5,2.5) -- (6,2.5) -- (6,2) -- (7,2) -- (7,2.5) -- (8,2.5) -- (8,1.5) -- (9,1.5) -- (9,2) -- (10,2) -- (10,1.5) -- (11,1.5) -- (11,2.5) -- (13,2.5) -- (13,0);
			
			\draw[dashed] (1,0.5) -- (13,0.5);
			\draw[dashed] (4,1.5) -- (13,1.5); 
			\draw[dashed] (5,2) -- (8,2);
			\draw[dashed] (3,0) -- (3,0.5);			
			\draw[dashed] (4,0) -- (4,0.5);
			\draw[dashed] (5,0) -- (5,1.5);
			\draw[dashed] (6,0) -- (6,2);
			\draw[dashed] (7,0) -- (7,2);
			\draw[dashed] (8,0) -- (8,1.5);
			\draw[dashed] (9,0) -- (9,2);
			\draw[dashed] (10,0) -- (10,2);
			\draw[dashed] (11,0) -- (11,2);
			
			\draw[thick, color=blue, pattern=north east lines, pattern color = blue] (1,0.5) rectangle (3,2.5);
			\draw[thick, color=red, pattern=north east lines, pattern color = red] (7,0.5) rectangle (11,1.5);
			\draw[thick, color=green, pattern=north east lines, pattern color = green] (11,0.5) rectangle (13,2.5);
			\draw[thick, color=yellow, pattern=north east lines, pattern color = yellow] (5,1.5) rectangle (6,2.5);
			\draw[thick, color=violet, pattern=north east lines, pattern color = violet] (4,0.5) rectangle (6,1.5);
		\end{tikzpicture}} \caption{}
\label{fig:zbend1} \end{subfigure}

 \caption{\textbf{(a)} The capacity profile along with the sets of lines $\Hs$ and $\Vs$ and some top-drawn jobs; \textbf{(b)} The jobs after processing; \protect \\
 }
\label{processingtopdrawn} 
\end{figure}
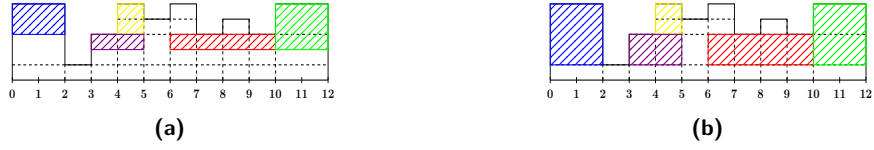

Note that $\R'_{\mathrm{large}}$ contains at most $m^4$ distinct types of jobs: the endpoints $s_{\js}$ and $t_{\js}$ can be chosen in ${m+1 \choose 2}\le m^2$ ways and the top and bottom edges of $R_{\js}$ must coincide with two lines from $\Hs$, which can be again chosen in ${m \choose 2}\le m^2$ ways. Let $U$ denote the number of types of job of the given instance and
let $\R'_{\mathrm{large}} = J_1\cup J_2 \cup \ldots \cup J_U$ be the decomposition of $\R'_{\mathrm{large}}$ into the $U$ distinct job types. 

We now formulate the configuration LP for this instance. Let $\Cs$ denote the set of all possible configurations of a round containing jobs from $\R'_{\mathrm{large}}$, drawn as top-drawn sets. For each $C \in \Cs$, we introduce a variable $x_C$, which stands for the number of rounds having configuration $C\in \Cs$. We write $J_k \triangleleft C$ if configuration $C$ contains a job from $J_k$ (note that $C$ can contain at most one job from $J_k$). Then the relaxed configuration LP and its dual (which contains a variable $y_k$ for each set $J_k$) are as follows.

\begin{align*}
\text{minimize} \quad \sum_{C\in \Cs} x_{C}&
&\text{maximize} \quad \sum_{k=1}^{U} |J_k|y_k& \\
\text{subject to} \quad \sum_{C:J_{k} \triangleleft C} x_C &\geq |J_k|, \quad k = 1,\ldots,U
&\text{subject to } \quad \sum_{k:J_{k} \triangleleft C} y_k &\leq 1, \quad C\in \Cs \\
x_{C} & \geq 0, \quad\, C\in \Cs
&y_k & \geq 0, \quad\, k =1, \dots, U
\end{align*}

The dual LP can be solved via the ellipsoid method with a suitable separation oracle. We interpret $y_k$ as the weight of each job in $J_k$. Given $(y_k)_{k\in \{1,\ldots,U\}}$, the separation problem asks whether there exists a configuration where jobs are drawn as top-drawn sets and the total weight of all the jobs in the configuration exceeds 1. For this, we invoke the following result of Bonsma et al.~\cite{BonsmaSW11}.

\begin{theorem} [\cite{BonsmaSW11}] \label{thm:dynamicprogramforufp} 
Given an instance of UFP with a set of jobs $J'$, the maximum-weight top-drawn subset of $J'$ can be computed in $O(nm^3)$ time.
\end{theorem}

Let $(x^*_C)_{C\in \Cs}$ be an optimal basic solution of the primal LP. 
By the rank lemma, there are at most
at most $U$ configurations $C$ for which $x^*_C$ is non-zero. 
For each non-zero $x^*_C$, we introduce $\lfloor x^*_C \rfloor$ rounds with configuration $C$, thus creating at most $8\cdot\opt_{UFP}$ rounds (due to \Cref{lemma:itspacking}). Now let $\R''_{\mathrm{large}}\subseteq \R'_{\mathrm{large}}$ be the large jobs that are yet to be packed and let $\omega''_{\mathrm{large}}$ be their clique number (and note that $\omega''_{\mathrm{large}}\le 8\cdot\opt_{UFP}$). In particular, a feasible solution to the configuration LP for the rectangles in $\R''_{\mathrm{large}}$ is to select one more round for each configuration $C$ with $x^*_C>0$. Therefore, we conclude that 
$\omega''_{\mathrm{large}} \le U \le m^4$ since there are at most $U$ configurations $C$ with $x^*_C>0$  and 
for each point, each configuration contains at most one rectangle covering this point.

Our strategy is to invoke the following theorem on $\R''_{\mathrm{large}}$.

\begin{theorem}[\cite{parinbar}] \label{thm:parinyathm}  Given a
set of rectangles with clique number $\omega$, in polynomial
time, we can compute a coloring of the rectangles using $O(\omega\log\omega)$
colors such that no two rectangles of the same color intersect. \end{theorem}

Thus, if $\omega''_{\mathrm{large}}=O(\log m)$ then we obtain an $O(\log\log m)$-approximation
as desired. However, it might be that $\omega''_{\mathrm{large}}$ is
larger. In that case, we partition $\R''_{\mathrm{large}}$ into $\omega''_{\mathrm{large}}/\log m$
sets, such that each of them has a clique size of $O(\log m)$.
\begin{lemma}
There is a randomized polynomial time algorithm that w.h.p.~computes
a partition $\R''_{\mathrm{large}}=\R_{1}\dot{\cup}...\dot{\cup}\R_{\omega''_{\mathrm{large}}/\log m}$
such that for each set $\R_{i}$, the corresponding clique size is
at most $O(\log m)$.
\end{lemma}
\begin{proof}
We split the rectangles $\R''_{\mathrm{large}}$ uniformly at random
into $\omega''_{\mathrm{large}}/\log m$ sets $\R_{1},\ldots,\R_{\omega''_{\mathrm{large}}/\log m}$. Thus
the expected clique size in each set $\R_{i}$ at any point $p$ under
the profile is at most $\log m$. Using the Chernoff bound, the probability
that the clique size at $p$ is more than $8\log m$ is at most $2^{-8\log m}=1/m^{8}$. As before, we draw the set of horizontal and vertical lines $\Hs$ and $\Vs$, respectively,  under the capacity profile, dividing the region underneath the profile into at most $m^2$ regions. Clearly, the clique number must be the same at all points inside any such region. Thus the probability that there exists a point $p$ under the capacity profile where the clique size is more than $8\log m$ is at most $m^{2}/m^{8}\le1/m^{6}$.
Hence using union bound, probability that clique size is more than
$8\log m$ at some point in some set $\R_{i}$ is at most $1/m^{2}$ (since $\omega''_{\mathrm{large}}\le m^4$).
\end{proof}
We apply Theorem~\ref{thm:parinyathm} to each set $\R_{i}$
separately and thus obtain a coloring with $O(\log m\log\log m)$
colors. Thus, for all sets $\R_{i}$ together we use at most $\frac{\omega''_{\mathrm{large}}}{\log m}O(\log m\log\log m)=O(\omega''_{\mathrm{large}}\log\log m)$
colors. We pack the jobs from each color class to a separate round
for our solution to \rufp. This yields
an $O(\log\log m)$-approximation, together with Theorem~\ref{thm:Elbassioni-small}. 
Since $m\le 2n-1$ after our preprocessing, our algorithms are also $O(\log \log n)$-approximation algorithms.

\begin{theorem} There exists a randomized $O(\log\log \min\{n,m\})$-approximation
algorithm for \rufp\ for general edge capacities. \end{theorem}


In order to obtain an algorithm for \rsap, we invoke the following
lemma due to \cite{MomkeW15} to each round of the computed solution
to \rufp.
\begin{lemma}[\cite{MomkeW15}]
\label{lem:UFP2SAP}Let $J'$ be the set of jobs packed in a feasible round for a given instance
of \rufp. Then in polynomial time we can partition $J'$ into  $O(1)$ sets and compute a height $h_{\js}$ for each job $\js\in J'$
such that each set yields a feasible round of \rsap.
\end{lemma}

This yields a solution to \rsap~with only $O(OPT_{UFP}\log\log m) \le O(OPT_{SAP}\log\log m)$
many rounds.

\begin{theorem} There exists a randomized $O(\log\log \min\{n,m\})$-approximation
algorithm for \rsap\ for general edge capacities.\end{theorem}

\subsection{An $O(\log \log \frac{1}{\delta})$-approximation algorithm with $(1+\delta)$-resource augmentation}

We show that if we are allowed a resource augmentation of a factor of $1+\delta$ for some $\delta>0$, we can get an $O(\log \log \frac{1}{\delta})$-approximation for both \rsap\ and \rufp. 
Consider \rufp\ first.

\begin{lemma}
\label{lem:Jpack}
Let $J^{(i)} := \{\js \in J  \mid \mathfrak{b}_{\js}\in [1/\delta^i,1/\delta^{i+1})\}$.  For packing jobs in $J^{(i)}$, it can be assumed that the capacity of each edge lies in the range $[1/\delta^i,2/\delta^{i+1})$.
\end{lemma}
Hence using the following theorem, we get a $O(\log \log \frac{1}{\delta})$-approximate solutions for each $J^{(i)}$,
which in particular uses at most $O(\opt_{UFP}\log\log\frac{1}{\delta})$ rounds.

\begin{theorem} [\cite{JahanjouKR17}]
There is a polynomial time $O(\log \log \frac{c_{\max}}{c_{\min}})$-approximation algorithm for \rufp.
\end{theorem}

Next we argue that we can combine the rounds computed for the sets $J^{(i)}$.
More precisely, we show that if we take one round from each set $J^{(0)},J^{(2)},J^{(4)},...$
and form their union, then they form a feasible round for the given
instance under $(1+\delta)$-resource augmentation. The same holds
if we take one round from each set $J^{(1)},J^{(3)},J^{(5)}, \dots$.

\begin{lemma}
\label{lem:combine}Take one computed round for each set $J^{(2k)}$
with $k\in\mathbb{N}$ or one computed round from each set $J^{(2k+1)}$
with $k\in\mathbb{N}$, and let $J'$ be their union. Then $J'$ is
a feasible round for the given instance of \rufp\ under $(1+\delta)$-resource
augmentation.
\end{lemma}

Thus, due to Lemma~\ref{lem:combine} we obtain a solution
with at most $O(\opt_{UFP}\log\log\frac{1}{\delta})$ rounds for the overall
instance. As earlier, we  take the given \rufp~solution 
 and apply Lemma~\ref{lem:UFP2SAP} to it, which yields a solution
to \rsap\ with at most $O(\opt_{SAP}\log\log\frac{1}{\delta})$ rounds.

\begin{theorem}
\label{thm:loglogapproxwithresourceaugmentation}
There exists an $O(\log\log\frac{1}{\delta})$-approximation algorithm
for \rsap\ and \rufp\ for general edge capacities and $(1+\delta)$-resource
augmentation. 
\end{theorem}

\section{Algorithms for the no-bottleneck-assumption}
\label{sec:roundsapandufpwithnba}

In this section, we present a $(16+\eps)$-approximation algorithm
for \rsap\  and a $12$-approximation
 for \rufp, both under the no-bottleneck-assumption (NBA).  

\subsection{Algorithm for \rsap}

For our algorithm for \rsap\ under NBA, we first
scale down all job demands and edge capacities so that $c_{\min}:=\min_{e\in E}c_{e}=1$.
Since the NBA holds, this implies that $d_{\js}\le1$ for each job
$\js\in J$. Then, we scale down all edge capacities to the nearest
power of 2, i.e., for each edge $e\in E$ we define a new rounded
capacity $c'_{e}:=2^{\left\lfloor \log c_{e}\right\rfloor }$. We
introduce horizontal lines whose $y$-coordinates are integral powers
of 2 (in contrast to the lines with uniform spacing that we used in Section~\ref{sec:Uniform_capacity}). Formally,
we define a set of lines $\mathcal{L}:=\{\ell_{2^k}|k\in\mathbb{N}\}$.
Let $OPT'_{SAP}$ denote the optimal solution for the rounded down capacities
$(c'_{e})_{e\in E}$ under the additional constraint that there must be no job
whose rectangle intersects a line in $\mathcal{L}$. 


We now show that given a valid \rsap\ packing $\Ps$ of a set of jobs $J'$ for the
edge capacities $(c_{e})_{e\in E}$, there exists a valid packing
of $J'$ into 4 rounds $\Rs^{1},\Rs^{2},\Rs^{3},\Rs^{4},$ under profile
$(c'_{e})_{e\in E}$ such that no job is sliced by a line in $\mathcal{L}$ (see \Cref{closedcor2_newnew} in \Cref{sec:figforrsap}). Let $J^{(i)}$ be the set of jobs with bottleneck capacity in  $[2^i,2^{i+1})$. For $i\ge 1$, we pack $J^{(i)}$ as follows:
\begin{enumerate}
\item Place each job $\js$ lying completely below $\ell_{2^{i}}$ into
$\Rs^{1}$ at height $h'_{\js}:=h_{\js}$.
(Note that here $h_{\js}$ denotes the height of the bottom edge
of $\js$ in $\Ps$.) 
\item Place each job $\js$ lying completely between $\ell_{2^{i}}$ and
$\ell_{3\cdot2^{i-1}}$ into $\Rs^{2}$ with $h'_{\js}:=h_{\js}-2^{i-1}$. 
\item Place each job $\js$ lying completely between $\ell_{3\cdot2^{i-1}}$
and $\ell_{2^{i+1}}$ into $\Rs^{3}$ with $h'_{\js}:=h_{\js}-2^{i}$. 
\item Place each job $\js$ sliced by $\ell_{2^{i}}$ into $\Rs^{4}$ with
$h'_{\js}:=2^{i}-d_{\js}$; and place each job $\js$ sliced by $\ell_{3\cdot2^{i-1}}$ into $\Rs^{4}$
with $h'_{\js}:=3\cdot2^{i-2}-d_{\js}$ for $i\ge2$ and $h'_{\js}:=1-d_{\js}$
for $i=1$.
\end{enumerate}
\noindent Finally, we pack $J^{(0)}$ as follows:
\begin{enumerate}
\item Place each job $\js$ lying completely below $\ell_{1}$ into $\Rs^{1}$
with $h'_{\js}:=h_{\js}$. 
\item Place each job $\js$ lying completely between $\ell_{1}$ and $\ell_{2}$
into $\Rs^{2}$ with $h'_{\js}:=h_{\js}-1$. 
\item Place each job $\js$ sliced by $\ell_{1}$ into $\Rs^{3}$ with $h'_{\js}:=1-d_{\js}$. 
\end{enumerate}
It can be checked easily that the above algorithm yields a feasible packing of $J'$ in which no job is sliced by a line in $\mathcal{L}$. Thus by losing at most a factor of 4 in our approximation ratio, we can convert any valid \rsap\ packing to a packing under the profile $(c'_{e})_{e\in E}$ that satisfies the latter constraint as well.

\begin{lemma} \label{lemma:factor4} We have that $OPT'_{SAP}\le4\cdot OPT_{SAP}$
\end{lemma}

We shall now obtain a valid packing of $J$ for the edge capacities
$(c'_{e})_{e\in E}$. Let $c'_{\max}:=\max_{e\in E}c'_{e}$ and for
each $i\in\{0,1,\ldots,\log c'_{\max}\}$ let $J^{(i)}\subseteq J$
denote the set of jobs with bottleneck capacity $2^{i}$ according
to $(c'_{e})_{e\in E}$. For each set $J^{(i)}$ we create a new (artificial)
instance with uniform edge capacities: in the instance for $J^{(0)}$
all edges have capacity 1, and for each $i\in\{1,2,\ldots,\log c'_{\max}\}$
in the instance for $J^{(i)}$ all edges have capacity $2^{i-1}$. For
each $i\in\{0,1,2,\ldots,\log c'_{\max}\}$ denote by $OPT^{(i)}$ the
number of rounds needed in the optimal solution to the instance for
$J^{(i)}$. Since in the solution $OPT'_{SAP}$ no rectangle is intersected
by a line in $\mathcal{L}$, for each set $J^{(i)}$ we can easily rearrange
the jobs in $J^{(i)}$ in $OPT'_{SAP}$ such that we obtain a solution
for $J^{(i)}$ with at most $2\cdot OPT'_{SAP}$ rounds.

\begin{lemma} \label{lemma:factor2loss} For each $i\in\{0,1,2,\ldots,\log c'_{\max}\}$
it holds that $OPT^{(i)}\le2\cdot OPT'_{SAP}$. \end{lemma}

For each set $J^{(i)}$ we invoke our asymptotic $(2+\eps)$-approximation
algorithm for \rsap\ for uniform edge capacities (see Section~\ref{sec:Uniform_capacity})
and obtain a solution $\Gamma^{(i)}$ which hence uses $\Gamma^{(i)}\le (2+\eps)\cdot OPT^{(i)}+O(1)\le (2+\eps)\cdot OPT'_{SAP}+O(1)$
rounds. Finally, we combine the solutions $\Gamma^{(i)}$ for all $i\in\{0,1,2,\ldots,\log c'_{\max}\}$
to one global solution of $J$. The key insight for this is that if
some edge $e$ has a (rounded) capacity of $c'_{e}=2^{k}$, then in
one round we can place the solution of one round of each of the solutions
$\Gamma^{(0)},\Gamma^{(1)},...,\Gamma^{(k-1)}$. 
\begin{lemma}
\label{lemma:mainpackinglemma}Given a solution $\Gamma^{(i)}$ for
the set $J^{(i)}$ for each $i\in\{0,1,2,\ldots,\log c'_{\max}\}$, in polynomial time we can construct a solution for $J$ and the edge
capacities $(c'_{e})_{e\in E}$ with at most $\max^{(i)}\{\Gamma^{(i)}\}$
many rounds.
\end{lemma}

Using the asymptotic $(2+\eps)$- (resp. absolute 3-) approximation for \rsap under uniform edge capacities (from Theorems \ref{theorem:twoasymptoticapprox} and \ref{thm:unisapabs}),  this yields the following theorem.

\begin{theorem} \label{theorem:16approxroundsap} For any $\eps>0$,
there exists an asymptotic $(16+\eps)$-approximation and an absolute 24-approximation algorithm for
\rsap under the NBA. \end{theorem}

\subsection{Algorithm for \rufp}

\label{sec:generalroundufp} In this section, we present a 12-approximation
for \rufp\ under NBA. In \rufp, it is not clear how to bootstrap
the algorithm for the uniform case as we did for \rsap, since in
the optimal solution it might not be possible to draw the jobs as non-overlapping
rectangles. Instead, our algorithm refines combinatorial properties
from \cite{ElbassioniGGKN12} to obtain an improved approximation
ratio. We will use the following algorithm due to~\cite{NomikosPZ01}
as a subroutine. 
\begin{theorem}[\cite{NomikosPZ01}]
\label{theorem:nomikostheorem}Given an instance of \rufp\ in which
all job demands are equal to 1 and all edge capacities are integral.
Then $OPT_{UFP}=r$ and in polynomial time we can compute a packing
into $OPT_{UFP}$ many rounds.
\end{theorem}

\noindent Via scaling, we assume that $c_{\min}=1$ and the demand of each
job is at most 1. Let $J_{\mathrm{large}}:=\{\js\in J\mid d_{\js}>1/2\}$
and  $J_{\mathrm{small}}:=J\setminus J_{\mathrm{large}}$. For
each job $\js\in J_{\mathrm{small}}$ we round up its demand to the
next larger power of $1/2$, i.e., we define its rounded demand $d'_{\js}:=2^{\left\lceil \log d_{\js}\right\rceil }$.
For each $i\in\mathbb{N}$, let $J^{(i)}$ denote the set of jobs whose
demands after rounding equal $\frac{1}{2^{i}}$, i.e., $J^{(i)}=\{\js\in J|d'_{\js}=\frac{1}{2^{i}}\}$.
For each edge $e$ and each $i\in\mathbb{N}$, we count how many jobs
in $J^{(i)}$ use $e$ and we define $n_{e,i}:=|\{\js\in J^{(i)}\mid e\in P_{\js}\}|$.
We partition each set $J^{(i)}$ into the sets $J'^{(i)}=\{\js\in J^{(i)}\mid\exists e\in P_{\js}:n_{e,i}<2r\}$
and $J''^{(i)}=J\setminus J'^{(i)}$. Let $n'_{e,i}:=|\{\js\in J'^{(i)}\mid e\in P_{\js}\}|$
and $n''_{e,i}:=|\{\js\in J''^{(i)}\mid e\in P_{\js}\}|$. Clearly, $n'_{e,i}<4r$ for each edge $e$ and each $i$.

First, we compute a packing for $J_{\mathrm{small}}$. For the (small)
jobs in the sets $J'^{(i)}$, we use a packing method that ensures
that inside each round we have at most one job from each set $J'^{(i)}$. 
Since $n'_{e,i}<4r$ for each edge $e$ and each $i$, this needs at most $4r$ rounds.  
Moreover, by a geometric sum argument this yields a valid packing inside each round (the job demands sum up to at most $1=c_{\min}$). For the
jobs in $J''^{(i)}$, we partition the available capacity inside each
round among the sets $J''^{(i)}$ and then invoke the algorithm due to
Theorem~\ref{theorem:nomikostheorem} for each set $J''^{(i)}$ separately, which also needs at most $4r$ rounds, and thus at most $8r$ rounds in total.

\begin{lemma} For the jobs in $J_{\mathrm{small}}$ we can find a packing into at most $8r$ rounds. \end{lemma} 
\begin{proof}
For packing jobs in $\bigcup_{i}J'^{(i)}$, we introduce a set
$\Ts$ of $4r$ rounds. For each $i$, we consider the jobs in $J'^{(i)}$
in non-decreasing order of their left endpoints and assign a job to
the first round in $\Ts$ where it does not overlap with any job from
$J'^{(i)}$ that have been assigned till now. Since $n'_{e,i}\le4r$,
such an assignment is always possible. Thus over any edge $e$ inside
any round of $\Ts$, at most 1 job from each set $J^{(i)}$ can be present.
Thus the load on edge $e$ is at most $\sum\limits_{i}\frac{1}{2^{i}}\le1\le c_{e}$, 
and hence this is a valid packing.

For packing the jobs in $\bigcup_{i}J''^{(i)}$, we introduce
a set $\Ss$ of $4r$ rounds. Consider a set $J''^{(i)}$. Inside each
of these rounds, to each edge $e$ we assign a capacity of $\frac{1}{2^{i}}\cdot\lfloor\frac{n_{e,i}}{2r}\rfloor$
to $J''^{(i)}$. The resulting congestion of any edge having non-zero
capacity is $\frac{\frac{1}{2^{i}}n''_{e,i}}{\frac{1}{2^{i}}\lfloor\frac{n_{e,i}}{2r}\rfloor}\le\frac{\frac{n_{e,i}}{2r}}{\lfloor\frac{n_{e,i}}{2r}\rfloor}\cdot2r\le4r$.
Also since each job in $J''^{(i)}$ has demand equal to $\frac{1}{2^{i}}$,
the assigned capacity of each edge is an integral multiple of the
demand. Thus using \Cref{theorem:nomikostheorem}, jobs in $J''^{(i)}$
can be packed into at most $4r$ rounds with these capacities. When
we do this procedure for each set $J''^{(i)}$, we obtain that the total
load on each edge $e$ inside any round is at most $\sum\limits_{i}\frac{1}{2^{i}}\lfloor\frac{n_{e,i}}{2r}\rfloor\le\sum\limits_{i}\frac{1}{2^{i+1}}\frac{n_{e,i}}{r}\le\frac{1}{r}\sum\limits _{\js\colon e\in P_{j}}d_{j}\le c_{e}$.
Thus this is a valid packing.

Therefore, we pack all jobs in $J_{\mathrm{small}}$ into $8r$ rounds. 
\end{proof}
For the jobs in $J_{\mathrm{large}}$ we round up their demands to
1 and the edge capacities to the respective nearest integer. This increases the congestion $r$
by at most a factor of 4. Then we invoke Theorem~\ref{theorem:nomikostheorem}.

\begin{lemma} 
\label{lem:ufp4r}
The jobs in $J_{\mathrm{large}}$ can be packed into
$4r$ rounds. \end{lemma} 
%
Overall, this yields a 12-approximation algorithm for \rufp\ under
the NBA.

\begin{theorem} \label{theorem:16approxroundufp-1}There is a polynomial
time 12-approximation algorithm for \rufp\ under the NBA. \end{theorem}

\section{Algorithms for \rtree}
Extending the results for \rufp, using results on path coloring \cite{erlebach2001complexity}
and multicommodity demand flow \cite{ChekuriMS07}, we obtain the
following results for \rtree~(see App. \ref{sec:roundufpontree}).

\begin{theorem} \label{thm:rtreetheorem} For \rtree, there exists
a polynomial-time asymptotic (resp. absolute) 5.1- (resp. 5.5-) approximation algorithm for uniform edge capacities
and an asymptotic (resp. absolute) 49- (resp. 55-) approximation algorithm for the general case under the NBA. \end{theorem}




\bibliography{biblio2}

\appendix
\section{Approximation Algorithms}
\label{sec:approx}
In this subsection, we define notions related to approximation algorithms. 

\begin{definition}[Approximation Guarantee]
For a minimization problem $\Pi$, an algorithm $\mathcal{A}$ has approximation guarantee of  $\alpha$ ($\alpha>1$), if $\mathcal{A}(I) \le \alpha~\opt(I)$ for all input instance $I$ of $\Pi$. 
 For a maximization problem $\Pi'$, an algorithm $\mathcal{A}$ has approximation guarantee of  $\alpha$ ($\alpha>1$), if $\opt(I) \le \alpha~\mathcal{A}(I)$ for all input instance $I$ of $\Pi'$. 
\end{definition}
This is also known as {\em absolute approximation guarantee}. There is another notion of approximation called asymptotic approximation which we define next.

\begin{definition}[Asymptotic Approximation Guarantee]
 For a minimization problem $\Pi$, an algorithm $\mathcal{A}$ has asymptotic approximation guarantee of $\alpha$ ($\alpha>1$), if $\mathcal{A}(I) \le \alpha~\opt(I) + o(\opt(I))$ for all input instance $I$ of $\Pi$. 
\end{definition}

\begin{definition}[Polynomial Time Approximation Scheme (PTAS)]
 A minimization problem $\Pi$ admits PTAS if for every constant $\eps>0$, there exists a $(1+\eps)$-approximation algorithm with running time $O(n^{f(1/\eps)})$, for any function $f$ that depends only on $\eps$.
\end{definition}

If the running time of a PTAS is $O(f(1/\eps)~n^c)$ for some function $f$ and a constant $c$ that is independent of $1/\eps$, we call it Efficient PTAS (EPTAS), If the running time of a PTAS is polynomial in both $n$ and $1/\eps$, we call it Fully PTAS  (FPTAS). 
Asymptotic analogue of PTAS, EPTAS, FPTAS are known as APTAS, AEPTAS, AFPTAS, respectively. 

\section{Omitted Proofs}
\label{sec:applowb}

\subsection{Proof of \Cref{lem:propertyofdummyrounds}}
We first state some inequalities on the job dimensions.

\begin{lemma} \label{lemma:inequalities} The following inequalities
hold:
\begin{enumerate}[(i)]
	\item $d_{a}>999\gamma$, $\forall a\in A$.
	\item $d_{b}>1001\gamma$, $\forall b\in B$.
	\item $d_{a'}>1000\gamma$, $\forall a'\in A'$.
	\item $d_{b'}>997\gamma$, $\forall b'\in B'$.
\end{enumerate}
\end{lemma}
\begin{proof}
\begin{enumerate}[(i)]
\item Follows from definition of job dimensions.
\item Let $b \in B$. Then $d_b \ge 999\gamma + 4\cdot (\rho^4 -q\rho^3 - q\rho^2 - q\rho + 8)$, as $i,j,k \le q$. Thus $d_b - 1001\gamma 
\ge 4\cdot(32^4q^4-q \cdot 32^3 q^3 - q \cdot 32^2 q^2- q \cdot 32 q+8)-2\cdot (32^4q^4+15)
\ge 4\cdot (15\cdot 32^3q^4 - 32^2q^3 - 32q^2) + 2 > 0$ (using the fact that $\rho=32q$ and $\gamma=\rho^4 + 15$).
\item Let $a'\in A'$. Then $d_{a'} \ge 1001\gamma - 4\cdot (q\rho^3 + 4)$. Thus $d_{a'} - 1000\gamma \ge 28.32^3q^4 - 1 > 0$.
\item Let $b' \in B'$. Then $d_{b'} \ge 1001\gamma - 4 \cdot (\rho^4 - \rho^3 - \rho^2 - \rho +8)$. Thus $d_{b'} - 997\gamma > (4\gamma - 4\cdot \rho^4) + 4\cdot (\rho^3 + \rho^2 + \rho -8) >0$.
\end{enumerate}
\end{proof}

Thus the demand of any job in $A\cup A'\cup B\cup B'$ is at least $997\gamma$. Also recall that the demand of any dummy job is $2997\gamma$. Suppose a round $\Rs$ contains a dummy job. Consider the leftmost edge $e_1$. If $\Rs$ contains at least 2 jobs from $A\cup B$, the total sum of demands of the jobs inside $\Rs$ over $e_1$ would be at least $2997\gamma + 2\cdot 997\gamma > 4000\gamma = c^{*}$, a contradiction. Similarly $\Rs$ can contain at most 1 job from $A'\cup B'$. Thus any round containing a dummy job can contain at most 1 job each from $A\cup B$ and $A'\cup B'$.

\subsection{Proof of \Cref{lem:maxnumberofjobs}}
\label{sec:proofofmaxnumberofjobs}
From \Cref{lem:propertyofdummyrounds}, a round containing a dummy job can contain at most 3 jobs. Suppose a round $\Rs$ does not contain a dummy job. If it has at least 5 jobs from $A\cup B$, the load on the leftmost edge $e_1$ inside round $\Rs$ would be at least $5\cdot 997\gamma > 4000\gamma = c^{*}$, a contradiction. Similarly since every job in $A'\cup B'$ passes through the rightmost edge $e_m$, there can be at most 4 such jobs inside $\Rs$. Thus $\Rs$ can contain at most 8 jobs.

\subsection{Proof of \Cref{lem:propertyofnicerounds}}
The if part follows directly since the 8 jobs corresponding to a triplet $\tau_{l}=(x_{i},y_{j},z_{k})\in\Ts$
(4 from $A\cup B$ and their peers) can be packed together inside
one round as shown in Figure \ref{fig:Awellpackedbin_new}.

\begin{figure}[h]
\centering
\captionsetup{justification=centering}
\begin{tikzpicture}[scale=0.7, transform shape]
\draw[thick] (0,0) rectangle (8,4);

\draw[thick, fill=gray] (0,0) rectangle (3.6,1.15);
\draw[thick] (1.8,0.575) node {\small $b_l$};
\draw[thick, fill = gray] (0,1.15) rectangle (3.7,2.2);
\draw[thick] (1.85,1.675) node {\small $a_{Z,k}$};
\draw[thick, fill = gray] (0,2.2) rectangle (3.8,3.15);
\draw[thick] (1.9,2.675) node {\small $a_{Y,j}$};
\draw[thick, fill = gray] (0,3.15) rectangle (3.9,4);
\draw[thick] (1.95,3.575) node {\small $a_{X,i}$};

\draw[thick, fill = gray] (8,4) rectangle (3.9,2.85);
\draw[thick] (5.95,3.425) node {\small $a'_{X,i}$};
\draw[thick, fill = gray] (8,2.85) rectangle (3.8,1.8);
\draw[thick] (5.9,2.325) node {\small $a'_{Y,j}$};
\draw[thick, fill = gray] (8,1.8) rectangle (3.7,0.85);
\draw[thick] (5.85,1.325) node {\small $a'_{Z,k}$};
\draw[thick, fill = gray] (8,0.85) rectangle (3.6,0);
\draw[thick] (5.8,0.425) node {\small $b'_l$};

\draw[thick] (8,-0.3) node {\footnotesize $40000\gamma$};
\draw[thick] (-0.5,4) node {\footnotesize $4000\gamma$};
\draw[thick] (-0.2,-0.2) node {\footnotesize 0};
\end{tikzpicture}
\caption{A nice round corresponding to the triplet $\tau_l = (x_i,y_j,z_k)$}
\label{fig:Awellpackedbin_new}
\end{figure}
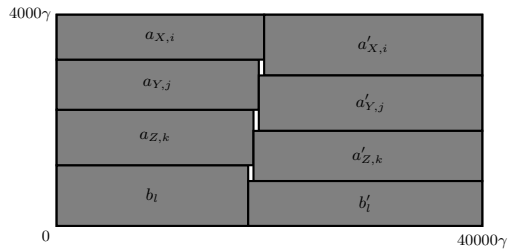

We begin the proof of the only if part by stating some properties of the job dimensions.

\begin{lemma}
\label{lemma:lemma1new} The following conditions hold: 
\begin{enumerate}[(i)]
\item $w_{a}+d_{a}=20999\gamma$, $\forall a\in A$.
\item $w_{a'}+d_{a'}=21001\gamma$, $\forall a'\in A'$. 
\item $w_{b}+d_{b}=20000\gamma$, $\forall b\in B$. 
\item $w_{b'}+d_{b'}=22000\gamma$, $\forall b'\in B'$. 
\end{enumerate}
\end{lemma}
\begin{proof}
Can be verified easily.
\end{proof}

\begin{lemma} \label{lemma:lemma2new} The following statements hold:
\begin{enumerate}
\item The demands of four jobs in $A\cup B$ (or $A'\cup B'$) sum
to $4000\gamma$ iff they correspond to a triplet.
\item The widths
of two jobs in $A\cup A'\cup B\cup B'$ sum to $40000\gamma$ iff
they are peers.
\end{enumerate}
\end{lemma}
\begin{proof}
\begin{enumerate}[(i)]
\item Follows directly from \Cref{lemma:woegingerlemma}.
\item The if part follows from the definition of the job dimensions. For the only if part, note that all the $5q$ integers in $\Vs$ are distinct. Also $x'_i < y'_j < z'_k < \tau'_l$, $\forall 1 \le i,j,k \le q$ and $l \le 2q$. It follows that all job widths are distinct. Thus for any $\js\in A\cup A' \cup B\cup B'$ having width $w_{\js}$, the unique job in $A\cup A' \cup B\cup B'$ having width $40000\gamma - w_{\js}$ is its peer. Hence the claim follows.
\end{enumerate}
\end{proof}

\begin{lemma} \label{lemma:startingandending} $s_{b'}<t_{a}$, $\forall a\in A$,
$b'\in B'$. \end{lemma}
\begin{proof}
Let $a\in A$ and $b'\in B'$. Then $t_{a} \ge 20000\gamma - 4\cdot (q\rho^3+4) \ge 20000\gamma - 4\gamma = 19996\gamma$, since $\gamma = \rho^4 + 15 \ge q\rho^3 + 4$. Also clearly $s_{b'} < 19001\gamma < 19996\gamma$. Thus we have $s_{b'} < t_a$.
\end{proof}

Now let $\Rs$ be a nice round. Since any \rsap\ packing is also a valid \rufp\ packing, it suffices
to prove the required result for \rufp\ packings. From the proof of \Cref{lem:maxnumberofjobs}, it follows that $\Rs$ must contain 4 jobs each from $A\cup B$ and $A'\cup B'$.

We begin by showing that $\Rs$ can contain at most 1 job from
$B$. Suppose at least 2 jobs from $B$ are present. From \Cref{lemma:inequalities}(i)
and (ii), the sum of demands of all the 4 jobs from $A\cup B$ is
then $>2\cdot1001\gamma+2\cdot999\gamma=4000\gamma = c^{*}$, a contradiction.
Now we show that at least one job from $B$ must be present. Suppose
not, then there are 4 jobs from $A$. \Cref{lemma:startingandending}
then implies that no job from $B'$ can be present. Thus there must
be 4 jobs from $A'$. But from \Cref{lemma:inequalities}(iii),
the sum of demands of these jobs from $A'$ would be $>4\cdot1000\gamma=4000\gamma=c^{*}$,
a contradiction. Hence, $\Rs$ must have exactly 1 job from $B$
and thus 3 jobs from $A$.

As shown above, $\Rs$ cannot contain 4 jobs from $A'$ and thus
must have at least 1 job from $B'$. We now show that exactly 1 job
from $B'$ must be present. \Cref{lemma:startingandending} implies
any job in $B'$ must share at least one edge with any job in $A$.
Let $e$ be such a common edge. Therefore, if at least 2 jobs from
$B'$ are present, sum of demands of all jobs inside $\Rs$ over $e$ would be $>3\cdot999\gamma+2\cdot997\gamma>4000\gamma = c^{*}$
(from \Cref{lemma:inequalities}), a contradiction. Thus $\Rs$
has exactly 1 job from $B'$ and thus 3 jobs from $A'$.

Now consider the leftmost edge $e_{1}$. Let the 4 jobs from $A\cup B$
lying over $e_{1}$, from top to bottom, be $\js_{1},\js_{2},\js_{3},\js_{4}$.
Similarly let the 4 jobs from $A'\cup B'$ lying over the rightmost
edge $e_{m}$ be $\js_{5},\js_{6},\js_{7},\js_{8}$ from top to bottom.
Then the following inequalities must hold: 
\begin{align}
\sum_{i=1}^{4}d_{\js_{i}}\le4000\gamma,\qquad\sum_{i=5}^{8}d_{\js_{i}}\le4000\gamma,\qquad\sum_{i=1}^{8}w_{\js_{i}}\le4\cdot40000\gamma=160000\gamma\label{eq: 200}
\end{align}

Adding the three inequalities, we get $\sum_{i=1}^{8}(d_{\js_{i}}+w_{\js_{i}})\le168000\gamma$.
Also \Cref{lemma:lemma1new} implies that $\sum_{i=1}^{8}(d_{\js_{i}}+w_{\js_{i}})=3\cdot20999\gamma+3\cdot21001\gamma+1\cdot20000\gamma+1\cdot22000\gamma=168000\gamma$.
Thus \eqref{eq: 200} must be satisfied with equality. \Cref{lemma:lemma2new}(i)
then implies that the four jobs from $A\cup B$ must correspond to
a triplet. Also since equality holds in \eqref{eq: 200}, \Cref{lemma:lemma2new}(ii)
implies that their corresponding peers must be present. Hence the result follows.

\subsection{Proof of \Cref{theorem:hardnessproofnew-1}}
Let $\opt$ denote either $\opt_{SAP}$ or $\opt_{UFP}$. From \Cref{lem:reduction-1}, $\opt \le 5q - 3\beta(q)$ whenever $\opt_{3DM}\ge \beta(q)$, and from \Cref{lem:reduction-2}, $\opt \ge (5q - 3\beta(q)) + \frac{1}{7}(\beta(q)-\alpha(q))$ whenever $\opt_{3DM} \le \alpha(q)$. Let $\delta_0 = \frac{1}{7}\frac{\beta(q)-\alpha(q)}{5q-3\beta(q)}$. Suppose there exists a polynomial time algorithm $\As$ for \rsap or \rufp and a constant $C$ such that for instances with $\opt > C$, $\As$ returns a packing using at most $(1+\delta_0)\opt$ rounds. Then for any corresponding instance $K$ of \sbmdm, we could distinguish whether $\opt_{3DM}<\alpha (q)$ or $\opt_{3DM}\ge \beta(q)$ in polynomial time, contradicting \Cref{theorem:chlebikhardness_new}. For $\alpha(q)= \lfloor 0.9690082645q \rfloor$ and $\beta(q) = \lceil 0.979338843q \rceil$, a simple calculation will show that $\delta_0 > 1/1398$.

\subsection{Proof of Lemma \ref{lem:normal}}
\label{sec:normalizedpacking}
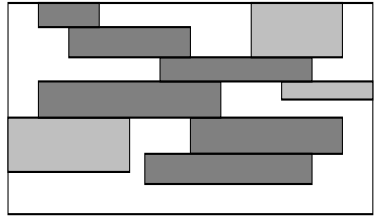
\begin{figure}[h]
\centering
\captionsetup{justification=centering}
\scalebox{0.5}[1]{
\begin{tikzpicture}[scale=0.8, transform shape]
\draw[thick] (0,0) rectangle (12,3.5);
\draw[thick, fill=lightgray] (0,0.7) rectangle (4,1.6);
\draw[thick, fill=gray] (6,1) rectangle (11,1.6);
\draw[thick, fill=gray] (1,1.6) rectangle (7,2.2);
\draw[thick, fill=lightgray] (9,1.9) rectangle (12,2.2);
\draw[thick, fill=gray] (5,2.2) rectangle (10,2.6);
\draw[thick, fill=gray] (2,2.6) rectangle (6,3.1);
\draw[thick, fill=gray] (1,3.1) rectangle (3,3.5);
\draw[thick, fill=lightgray] (8,2.6) rectangle (11,3.5);
\draw[thick, fill=gray] (4.5,0.5) rectangle (10,1);




\end{tikzpicture}}
\caption{A normalized packing}
\label{fig:chain}
\end{figure}
We sort the jobs in non-increasing order of their $h_{\js}+d_{\js}$
values and push them up until they either touch the capacity profile
or the bottom of some other job.

\subsection{Proof of Lemma \ref{lemma:maindynamicprogramming}}
We first consider \rsap. First we guess the value of $\opt_{SAP}$ as $\kappa$, where $\kappa \in \{1, 2, \dots, n\}$. Each DP cell consists of the following  $(\kappa +2)$  attributes. 
\begin{itemize}
  \item an edge $e \in E$,
  \item a function $f_{e}$ that assigns a round to each job passing through $e$,
  \item $\kappa$ functions $g^{1}_{e}$, $g^{2}_{e}$, \ldots, $g^{\kappa}_{e}$, one for each round, where $g^{i}_{e}$ assigns the vertical location $h_{\js}$ to each job $\js$ assigned to the $i^{\text{th}}$ round by $f_{e}$, $\forall i \in \{1,2, \ldots, \kappa\}$.  
\end{itemize}
 
For each edge $e_k$, let $J^{e_k} := \{\js \in J' \mid s_{\js}<k\}$, i.e., the set of jobs that either pass through or end before $e_k$. We define DP($e, f_{e}, g^{1}_{e}, g^{2}_{e}, \ldots, g^{\kappa}_{e}$) = 1 if and only if there exists a valid packing of $J^e$ using $\kappa$ rounds such that the positions of all jobs in $J^e$ are exactly the same as those assigned by the functions $f_{e}$, $g^{1}_{e}$, $g^{2}_{e}$, \ldots, $g^{\kappa}_{e}$. Thus the recurrence for the DP is given by DP($e_j, f_{e_j}, g^{1}_{e_j}, g^{2}_{e_j}, \ldots, g^{\kappa}_{e_j}$) = 1 if there exist functions $f_{e_{j-1}}$, $g^{1}_{e_{j-1}}$, $g^{2}_{e_{j-1}}$, \ldots, $g^{\kappa}_{e_{j-1}}$ such that DP($e_{j-1}, f_{e_{j-1}}, g^{1}_{e_{j-1}}, g^{2}_{e_{j-1}}, \ldots, g^{\kappa}_{e_{j-1}}$) = 1 and $(f_{e_j}, g^{1}_{e_{j}}, g^{2}_{e_{j}}, \ldots, g^{\kappa}_{e_{j}})$ and $(f_{e_{j-1}}$, $g^{1}_{e_{j-1}}$, $g^{2}_{e_{j-1}}$, \ldots, $g^{\kappa}_{e_{j-1}})$ are \textit{consistent} with each other, $\forall j\in \{2,\ldots,m\}$. Here by consistency, we mean that any job $\js$ that passes through both $e_j$ and $e_{j-1}$ must be assigned the same round (say the $i^{\text{th}}$ round) by $f_{e_j}$ and $f_{e_{j-1}}$ and the same value of $h_{\js}$ by $g^{i}_{e_{j}}$ and $g^{i}_{e_{j-1}}$.

Finally, we bound the running time. Since each job can be assigned to any of the $\kappa$ rounds and have at most $|\Hs'|$ distinct positions inside a round, the number of DP cells per edge is bounded by $(\kappa |\Hs'|)^{\omega}$. Also determining each DP entry requires visiting all DP cells corresponding to the edge on the immediate left of the current edge. Thus the time required to determine whether $J$ can be packed using $\kappa$ rounds is bounded by $m\cdot (\kappa |\Hs'|)^{\omega}\cdot (\kappa |\Hs'|)^{\omega} \le (n|\Hs'|)^{O(\omega)}$. 

For \rufp, the positions of the jobs inside a round does not matter. Thus we simply have a 2-cell DP, consisting of an edge $e$ and a function $f_e$ that allocates a round to each job passing through $e$. The recurrence is given by DP($e_j,f_{e_j}$) = 1 if there exists $f_{e_{j-1}}$ such that DP($e_{j-1},f_{e_{j-1}}$) = 1 and $f_{e_j}$ and $f_{e_{j-1}}$ are consistent with each other. Clearly the running time is bounded by $n^{O(\omega)}$.



\subsection{Proof of \Cref{lem:Jpack}}
Consider any edge $e\in E$ inside any round $\Rs$. Let $e_L$ and $e_R$ be the first edges on the left and right of $e$ respectively that have capacity at most $1/\delta^{i+1}$. Observe that any job in $J^{(i)}$ must pass through at least one edge having capacity at most $1/\delta^{i+1}$. Thus the load on edge $e$ inside round $\Rs$ is at most the sum of the loads on edges $e_L$ and $e_R$, which is at most $2/\delta^{i+1}$. Also since the bottleneck capacity of any job in $J^{(i)}$ is at least $1/\delta^i$, any edge with capacity less than $1/\delta^i$ can be contracted (i.e. its capacity can be made 0).

\subsection{Proof of \Cref{lem:combine}}
\label{sec:proofofloglog}

Let $\Rs^{(2k)}$ be the computed round for each set $J^{(2k)}$, for $k\in \mathbb{N}$. Suppose we are allowed a resource augmentation of $1+\gamma$, for some $\gamma > 0$. Now since the bottleneck capacity of any job in $J^{(i)}$ is at least $\frac{1}{\delta^i}$, on resource augmentation, the capacity of any such edge would increase by at least $\frac{\gamma}{\delta^i}$. Given the packing inside round $\Rs^i$, we shift up the height $h_{\js}$ of each job $\js$ by $\frac{\gamma}{\delta^i}$ and place it in a round $\Rs$ of the original capacity profile.

Now the jobs from $J^{(i-2)}$ do not go above $\frac{2}{\delta^{i-1}}+\frac{\gamma}{\delta^{i-2}}$. We choose $\gamma$ such that $\frac{\gamma}{\delta^i} \ge \frac{2}{\delta^{i-1}}+\frac{\gamma}{\delta^{i-2}}$, from which we get $\gamma \ge \frac{2\delta}{1-\delta^2}=O(\delta)$. Thus one round from each set $J^{(2k)}$ can be packed together. Similarly one round from each set $J^{(2k+1)}$ can be packed together.

\subsection{Proof of \Cref{lemma:factor4}}
\label{sec:figforrsap}
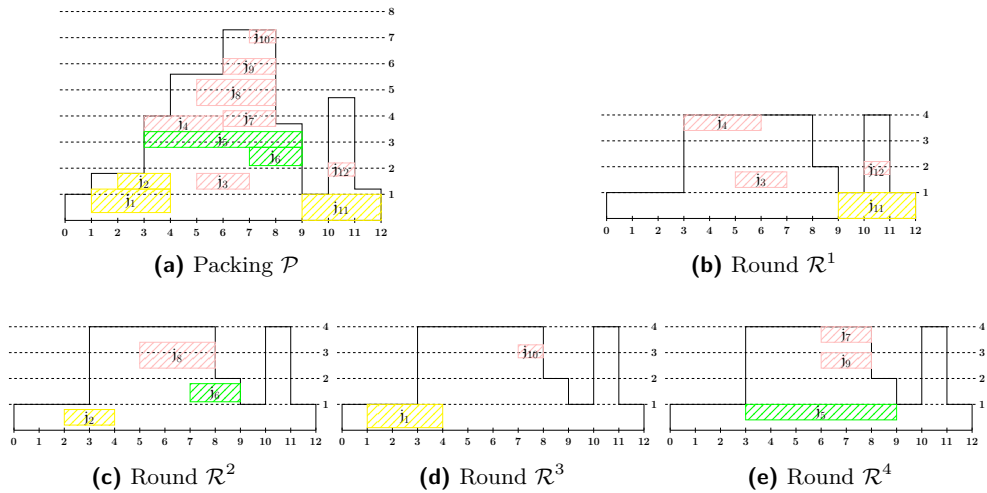
\begin{figure}[h]
\captionsetup{justification=centering} \captionsetup[subfigure]{justification=centering}
\begin{subfigure}[b]{.5\textwidth} \centering \resizebox{4.5cm}{3.1cm}{ \begin{tikzpicture}
			\draw[thick] (1,0) -- (13,0);
			\foreach \x in {1,...,13}
				\draw[thick] (\x,0.1) -- (\x,-0.1);
			\foreach \x in {0,...,12}
				\draw[thick] (\x+1,-0.45) node {\Large \textbf{\x}};
				
			\draw[thick] (1,0) -- (1,1) -- (2,1) -- (2,1.8) -- (4,1.8) -- (4,4) -- (5,4) -- (5,5.6) -- (7,5.6) -- (7,7.3) -- (9,7.3) -- (9,3.7) -- (10,3.7) -- (10,1) -- (11,1) -- (11,4.7) -- (12,4.7) -- (12,1.2) -- (13,1.2) -- (13,0);
			\foreach \x in {1,...,8}
				\draw[dashed] (0.8,\x) -- (13.2,\x);
				\foreach \x in {1,...,8}
				\draw[thick] (13.4,\x) node {\Large \textbf{\x}};
			
			\draw[thick, color=yellow, pattern=north east lines, pattern color = yellow] (2,0.3) rectangle (5,1.2);
			\draw[thick] (3.5,0.75) node {\LARGE \textbf{$\js_1$}};
			\draw[thick, color=yellow, pattern=north east lines, pattern color = yellow] (3,1.2) rectangle (5,1.8);
			\draw[thick] (4,1.5) node {\LARGE \textbf{$\js_2$}};
			\draw[thick, color=pink, pattern=north east lines, pattern color = pink] (6,1.2) rectangle (8,1.8);
			\draw[thick] (7,1.5) node {\LARGE \textbf{$\js_3$}};
			\draw[thick, color=pink, pattern=north east lines, pattern color = pink] (4,3.4) rectangle (7,4);
			\draw[thick] (5.5,3.7) node {\LARGE \textbf{$\js_4$}};
			\draw[thick, color=green, pattern=north east lines, pattern color = green] (4,2.8) rectangle (10,3.4);
			\draw[thick] (7,3.1) node {\LARGE \textbf{$\js_5$}};
			\draw[thick, color=green, pattern=north east lines, pattern color = green] (8,2.1) rectangle (10,2.8);
			\draw[thick] (9,2.45) node {\LARGE \textbf{$\js_6$}};
			\draw[thick, color=pink, pattern=north east lines, pattern color = pink] (7,3.6) rectangle (9,4.2);
			\draw[thick] (8,3.9) node {\LARGE \textbf{$\js_7$}};
			\draw[thick, color=pink, pattern=north east lines, pattern color = pink] (6,4.4) rectangle (9,5.4);
			\draw[thick] (7.5,4.9) node {\LARGE \textbf{$\js_8$}};
			\draw[thick, color=pink, pattern=north east lines, pattern color = pink] (7,5.6) rectangle (9,6.2);
			\draw[thick] (8,5.9) node {\LARGE \textbf{$\js_9$}};
			\draw[thick, color=pink, pattern=north east lines, pattern color = pink] (8,6.8) rectangle (9,7.3);
			\draw[thick] (8.5,7.05) node {\LARGE \textbf{$\js_{10}$}};
			\draw[thick, color=yellow, pattern=north east lines, pattern color = yellow] (10,0) rectangle (13,1);
			\draw[thick] (11.5,0.5) node {\LARGE \textbf{$\js_{11}$}};
			\draw[thick, color=pink, pattern=north east lines, pattern color = pink] (11,1.7) rectangle (12,2.2);
			\draw[thick] (11.5,1.95) node {\LARGE \textbf{$\js_{12}$}};
		\end{tikzpicture}} \caption{Packing $\Ps$}
\label{fig:zbend1} \end{subfigure} \begin{subfigure}[b]{.5\textwidth}
\centering \resizebox{4.4cm}{1.7cm}{ \begin{tikzpicture}
			\draw[thick] (1,0) -- (13,0);
			\foreach \x in {1,...,13}
				\draw[thick] (\x,0.1) -- (\x,-0.1);
			\foreach \x in {0,...,12}
				\draw[thick] (\x+1,-0.45) node {\Large \textbf{\x}};
				
				\draw[thick] (1,0) -- (1,1) -- (4,1) -- (4,4) -- (9,4) -- (9,2) -- (10,2) -- (10,1) -- (11,1) -- (11,4) -- (12,4) -- (12,1) -- (13,1) -- (13,0);
				\foreach \x in {1,...,4}
				\draw[dashed] (0.8,\x) -- (13.2,\x);
				\foreach \x in {1,...,4}
				\draw[thick] (13.4,\x) node {\Large \textbf{\x}};
				
				\draw[thick, color=pink, pattern=north east lines, pattern color = pink] (6,1.2) rectangle (8,1.8);
				\draw[thick] (7,1.5) node {\LARGE \textbf{$\js_3$}};
				\draw[thick, color=pink, pattern=north east lines, pattern color = pink] (4,3.4) rectangle (7,4);
			\draw[thick] (5.5,3.7) node {\LARGE \textbf{$\js_4$}};
				\draw[thick, color=yellow, pattern=north east lines, pattern color = yellow] (10,0) rectangle (13,1);
				\draw[thick] (11.5,0.5) node {\LARGE \textbf{$\js_{11}$}};
			\draw[thick, color=pink, pattern=north east lines, pattern color = pink] (11,1.7) rectangle (12,2.2);
			\draw[thick] (11.5,1.9) node {\LARGE \textbf{$\js_{12}$}};
		\end{tikzpicture}} \caption{Round $\Rs^{1}$}
\label{fig:zbend1} \end{subfigure}

\vspace{0.5cm}
 \centering \begin{subfigure}[b]{.3\textwidth} 
 \resizebox{4.3cm}{1.7cm}{ \begin{tikzpicture}
			\draw[thick] (1,0) -- (13,0);
			\foreach \x in {1,...,13}
				\draw[thick] (\x,0.1) -- (\x,-0.1);
			\foreach \x in {0,...,12}
				\draw[thick] (\x+1,-0.45) node {\Large \textbf{\x}};
				
				\draw[thick] (1,0) -- (1,1) -- (4,1) -- (4,4) -- (9,4) -- (9,2) -- (10,2) -- (10,1) -- (11,1) -- (11,4) -- (12,4) -- (12,1) -- (13,1) -- (13,0);
				\foreach \x in {1,...,4}
				\draw[dashed] (0.8,\x) -- (13.2,\x);
				\foreach \x in {1,...,4}
				\draw[thick] (13.4,\x) node {\Large \textbf{\x}};
				
				\draw[thick, color=yellow, pattern=north east lines, pattern color = yellow] (3,0.2) rectangle (5,0.8);
				\draw[thick] (4,0.5) node {\LARGE \textbf{$\js_2$}};
				\draw[thick, color=green, pattern=north east lines, pattern color = green] (8,1.1) rectangle (10,1.8);
				\draw[thick] (9,1.45) node {\LARGE \textbf{$\js_6$}};
				\draw[thick, color=pink, pattern=north east lines, pattern color = pink] (6,2.4) rectangle (9,3.4);
				\draw[thick] (7.5,2.9) node {\LARGE \textbf{$\js_8$}};
		\end{tikzpicture}} \caption{Round $\Rs^{2}$}
\label{fig:zbend1} \end{subfigure} 
 \begin{subfigure}[b]{.3\textwidth} \centering \resizebox{4.3cm}{1.7cm}{ \begin{tikzpicture}
			\draw[thick] (1,0) -- (13,0);
			\foreach \x in {1,...,13}
				\draw[thick] (\x,0.1) -- (\x,-0.1);
			\foreach \x in {0,...,12}
				\draw[thick] (\x+1,-0.45) node {\Large \textbf{\x}};
				
				\draw[thick] (1,0) -- (1,1) -- (4,1) -- (4,4) -- (9,4) -- (9,2) -- (10,2) -- (10,1) -- (11,1) -- (11,4) -- (12,4) -- (12,1) -- (13,1) -- (13,0);
				\foreach \x in {1,...,4}
				\draw[dashed] (0.8,\x) -- (13.2,\x);
				\foreach \x in {1,...,4}
				\draw[thick] (13.4,\x) node {\Large \textbf{\x}};
				
				\draw[thick, color=yellow, pattern=north east lines, pattern color = yellow] (2,0.1) rectangle (5,1);
				\draw[thick] (3.5,0.55) node {\LARGE \textbf{$\js_1$}};
				
				\draw[thick, color=pink, pattern=north east lines, pattern color = pink] (8,2.8) rectangle (9,3.3);
				\draw[thick] (8.5,3.05) node {\LARGE \textbf{$\js_{10}$}};
	
		\end{tikzpicture}} \caption{Round $\Rs^{3}$}
\label{fig:zbend1} \end{subfigure} 
 \begin{subfigure}[b]{.3\textwidth} \centering \resizebox{4.3cm}{1.7cm}{ \begin{tikzpicture}
			\draw[thick] (1,0) -- (13,0);
			\foreach \x in {1,...,13}
				\draw[thick] (\x,0.1) -- (\x,-0.1);
			\foreach \x in {0,...,12}
				\draw[thick] (\x+1,-0.45) node {\Large \textbf{\x}};
				
				\draw[thick] (1,0) -- (1,1) -- (4,1) -- (4,4) -- (9,4) -- (9,2) -- (10,2) -- (10,1) -- (11,1) -- (11,4) -- (12,4) -- (12,1) -- (13,1) -- (13,0);
				\foreach \x in {1,...,4}
				\draw[dashed] (0.8,\x) -- (13.2,\x);
				\foreach \x in {1,...,4}
				\draw[thick] (13.4,\x) node {\Large \textbf{\x}};
				
				\draw[thick, color=green, pattern=north east lines, pattern color = green] (4,0.4) rectangle (10,1);
				\draw[thick] (7,0.7) node {\LARGE \textbf{$\js_5$}};
				\draw[thick, color=pink, pattern=north east lines, pattern color = pink] (7,2.4) rectangle (9,3);
				\draw[thick] (8,2.7) node {\LARGE \textbf{$\js_9$}};
				\draw[thick, color=pink, pattern=north east lines, pattern color = pink] (7,3.4) rectangle (9,4);
				\draw[thick] (8,3.7) node {\LARGE \textbf{$\js_7$}};
	
		\end{tikzpicture}} \caption{Round $\Rs^{4}$}
\label{fig:zbend1} \end{subfigure} \caption{Figure for \Cref{lemma:factor4}. The yellow, green and pink jobs
belong to $T^{(0)}$, $T^{(1)}$ and $T^{(2)}$, respectively.}
\label{closedcor2_newnew} 
\end{figure}

Given a valid \rsap\ packing $\Ps$ of a set of jobs $J'$ for the
given edge capacities $(c_{e})_{e\in E}$, we construct a valid packing
of $J'$ into 4 rounds $\Rs^{1},\Rs^{2},\Rs^{3},\Rs^{4},$ under profile
$(c'_{e})_{e\in E}$ such that no job is sliced by $\mathcal{L}$.
Let $J^{(i)}:=\{\js\in T\mid2^{i}\le\mathfrak{b}_{\js}<2^{i+1}\}$,
$\forall i\in\{0,1,\ldots,\lfloor\log c_{\max}\rfloor\}$. Note that
jobs in $J^{(i)}$ will have bottleneck capacity equal to $2^{i}$ in
$(c'_{e})_{e\in E}$. For $i\in\{1,\ldots,\lfloor\log c_{\max}\rfloor\}$,
we pack $J^{(i)}$ as follows: 
\begin{itemize}
\item Place each job $\js$ lying completely below $\ell_{2^{i}}$ into
$\Rs^{1}$ at height $h'_{\js}:=h_{\js}$.
(Note that here $h_{\js}$ denotes the height of the bottom edge
of $\js$ in $\Ps$.) 
\item Place each job $\js$ lying completely between $\ell_{2^{i}}$ and
$\ell_{3\cdot2^{i-1}}$ into $\Rs^{2}$ with $h'_{\js}:=h_{\js}-2^{i-1}$. 
\item Place each job $\js$ lying completely between $\ell_{3\cdot2^{i-1}}$
and $\ell_{2^{i+1}}$ into $\Rs^{3}$ with $h'_{\js}:=h_{\js}-2^{i}$. 
\item Place each job $\js$ sliced by $\ell_{2^{i}}$ into $\Rs^{4}$ with
$h'_{\js}:=2^{i}-d_{\js}$. 
\item Place each job $\js$ sliced by $\ell_{3\cdot2^{i-1}}$ into $\Rs^{4}$
with $h'_{\js}:=3\cdot2^{i-2}-d_{\js}$ for $i\ge2$ and $h'_{\js}:=1-d_{\js}$
for $i=1$. 
\end{itemize}
Note that the jobs allocated to $\Rs^{1}$ are placed at the same
vertical height as in $\Ps$. For $\Rs^{2}$, we have maintained the
invariant that all jobs placed between $\ell_{2^{i-1}}$ and $\ell_{2^{i}}$
lied completely between $\ell_{2^{i}}$ and $\ell_{3.2^{i-1}}$ and
had bottleneck capacity in the interval $[2^{i},2^{i+1})$ in $\Ps$.
Similarly for $\Rs_{3}$, jobs placed between $\ell_{2^{i-1}}$ and
$\ell_{2^{i}}$ lied completely between $\ell_{3.2^{i-1}}$ and $\ell_{2^{i+1}}$
and had bottleneck capacity in the interval $[2^{i},2^{i+1})$ in
$\Ps$. Observe that the region below $\ell_{1}$ is currently empty
in both $\Rs^{2}$ and $\Rs^{3}$ (this will be utilized for packing
jobs in $J^{(0)}$). Finally, $\Rs^{4}$ contains the jobs sliced by
$\ell_{2^{i}}$ and $\ell_{3\cdot2^{i-1}}$.

As we have ensured that jobs from $J^{(i)}$ have been placed only in
the region between the slicing lines $\ell_{2^{i-1}}$ and $\ell_{2^{i}}$
in $\Rs^{2}$ and $\Rs^{3}$, there is no overlap between jobs from
different $J^{(i)}$'s in $\Rs^{2}$ and $\Rs^{3}$. Also since the
demand of each job is at most 1, there is no overlapping of jobs from
different $J^{(i)}$'s inside $\Rs^{4}$. Now we pack the jobs in $J^{(0)}$
(jobs with bottleneck capacity less than 2), as follows: 
\begin{itemize}
\item Place each job $\js$ lying completely below $\ell_{1}$ into $\Rs^{1}$
with $h'_{\js}:=h_{\js}$. 
\item Place each job $\js$ lying completely between $\ell_{1}$ and $\ell_{2}$
into $\Rs^{2}$ with $h'_{\js}:=h_{\js}-1$. 
\item Place each job $\js$ sliced by $\ell_{1}$ into $\Rs^{3}$ with $h'_{\js}:=1-d_{\js}$. 
\end{itemize}
We continue to maintain the same invariant for $\Rs^{1}$. For $\Rs^{2}$
and $\Rs^{3}$, we used the empty region below $\ell_{1}$ to pack
jobs from $J^{(0)}$. Note that no job is sliced by $\Ls$ in the packing
obtained. Hence the lemma holds.

\subsection{Proof of Lemma \ref{lemma:factor2loss}}

For $i=0$, $\opt^{(0)} \le \opt'_{SAP}$ from definition and we are done. For $i \ge 1$, given a valid packing under uniform capacity profiles of height $2^i$ with no job sliced by $\mathcal{L}$ (call it type 1), we obtain a valid packing under uniform capacity profiles of height $2^{i-1}$ (type 2) using at most twice the number of rounds, thus proving that $\opt^{(i)}\le 2\cdot \opt'_{SAP}$. For each round of type 1, introduce 2 rounds of type 2, denoted by $\Rs^1$ and $\Rs^2$. Note that since $\ell_{2^{i-1}} \in \mathcal{L}$, no job is sliced by $\ell_{2^{i-1}}$ in any round of type 1. We pack all jobs lying above $\ell_{2^{i-1}}$ into $\Rs^1$ and all jobs below $\ell_{2^{i-1}}$ into $\Rs^2$.

\subsection{Proof of \Cref{lemma:mainpackinglemma}}

Let $\Rs^i \in \Gamma_i$ be the round chosen for each $i$. The crucial observation is that for any $\js \in J_i$, $P_{\js}$ only contains edges of capacity at least $2^i$. Thus each $\Rs^i$, $1 \le i \le \log c'_{\max}$ can be placed between the slicing lines $\ell_{2^{i-1}}$ and $\ell_{2^{i}}$ under the profile $I'$. Finally we place $\Rs^0$ between $\ell_0$ and $\ell_1$. Observe that no job is sliced by $\mathcal{L}$ in this packing generated.    

\subsection{Proof of \Cref{lem:ufp4r}}

We round up the demand of each job in $J_{\mathrm{large}}$ to $1$
and round down the capacity of each edge to the nearest integer. Thus
the congestion increases by at most a factor of 4. Hence, by applying
\Cref{theorem:nomikostheorem}, jobs in $J_{\mathrm{large}}$ can be packed into
at most $4r$ bins. 

\section{\rufp\ on Trees}
\label{sec:roundufpontree}
\subsection{Uniform capacities}
\label{subsec:roundtreeuniform}
First, we consider the \rtree\ problem where all edges of $G_{\text{tree}}$ have the same capacity $\cs$. Let $\Zs_L := \{\js \in J\mid d_{\js} > \cs/2 \}$ and $\Zs_S := J\setminus \Zs_L$.

For the jobs in $\Zs_L$, any two jobs sharing at least one common edge must be placed in different \textit{tree-rounds}. Thus the problem reduces to the \pcolor\ problem on trees and we have the following result due to Erlebach and Jansen \cite{erlebach2001complexity} (for asymptotic approximation ratio) and Raghavan and Upfal \cite{raghavan1994efficient} (for absolute approximation ratio).


\begin{lemma}
\label{lemma:almostopt}
There exists a polynomial time asymptotic (resp. absolute) 1.1- (resp. 1.5-) approximation algorithm for packing jobs in $\Zs_L$.
\end{lemma}

Now we consider packing of jobs in  $\Zs_S$.  First, we fix any vertex $v_{\text{root}}$ as the root of $G_{\text{tree}}$. For any $\js \in J$, let $\theta_{\js}$ denote the least common ancestor of $v_{s_{\js}}$ and $v_{t_{\js}}$. For any two vertices $u,v \in V(G_{\text{tree}})$, let $P_{u-v}$ denote the unique path starting at $u$ and ending at $v$. The \textit{level} of a vertex $v\in G_{\text{tree}}$ is defined as the number of edges present in $P_{v_{\text{root}}-v}$. We sort the jobs in $\Zs_S$ in non-decreasing order of the levels of $\theta_{\js}$ and apply the \ff\ algorithm, i.e., we assign $\js$ to the first tree-round in which it can be placed without violating the edge-capacities. Let $\Gamma$ denote the set of tree-rounds used.

\begin{lemma}
\label{lemma:badfactorfour}
$|\Gamma| \le 4r$.
\end{lemma}
\begin{proof}
For any job $\js$, let $e^{1}_{\js}$ and $e^{2}_{\js}$ denote the first edges on the paths $P_{\theta_{\js}-v_{s_{\js}}}$ and $P_{\theta_{\js}-v_{t_{\js}}}$, respectively. For any tree-round $\Rs \in \Gamma$ and edge $e\in E(G_{\text{tree}})$, let $l^{\Rs}_e$ denote the sum of demands of all jobs assigned to $\Rs$ that pass through $e$.  

Let $\Rs^1, \Rs^2, \ldots \Rs^{|\Gamma|}$ be the tree-rounds used to pack $\Zs_S$. Observe that only $l^{\Rs^i}_{e^{1}_{\js}}$ and $l^{\Rs^i}_{e^{2}_{\js}}$, $\forall 1\le i \le |\Gamma|$, need to be considered while determining the tree-round for $\js$, as we are considering jobs in $\Zs_S$ in non-decreasing order of the levels of $\theta_{\js}$.
Let $\js_{\text{f}}$ be the first job for which $\Rs^{|\Gamma|}$ is opened for the first time. Clearly then at least one of $l^{\Rs^i}_{e^{1}_{\js_{\text{f}}}}$ or $l^{\Rs^i}_{e^{2}_{\js_{\text{f}}}}$ must exceed $\cs/2$, $\forall i = 1,2,\ldots,|\Gamma|-1$. Thus $\frac{\cs}{2} < l^{\Rs^i}_{e^{1}_{\js_{\text{f}}}} + l^{\Rs^i}_{e^{2}_{\js_{\text{f}}}}$. Summing over all $i$ from 1 to $|\Gamma| -1$, we get $\frac{\cs}{2}(|\Gamma| -1) < l_{e^{1}_{\js_{\text{f}}}} + l_{e^{2}_{\js_{\text{f}}}} \le 2L$. Thus, $|\Gamma| < 4r+1$. Since. $|\Gamma|$ is an integer, we get $|\Gamma| \le 4r$.     
\end{proof}

Combining \Cref{lemma:almostopt} and \Cref{lemma:badfactorfour}, we have the following theorem.

\begin{theorem}
There exists a polynomial-time asymptotic (resp. absolute) 5.1- (resp. 5.5-) approximation algorithm for \rtree~with uniform edge capacities.
\end{theorem}

\subsection{Arbitrary capacities (under NBA)}
\label{subsec:roundtreearbi}
We now consider the \rtree\ problem. Chekuri et al.~\cite{ChekuriMS07}  proved the following result.

\begin{theorem}
\label{theorem:Chekuristheorem}
\cite{ChekuriMS07} If all job demands are equal to 1, then they can be packed into at most $4r$ \textit{tree-rounds}.
\end{theorem}

We state the following result which enables us to obtain an improved approximation algorithm.

\begin{lemma}
\label{lemma:treelemma}
Let $\Hs := \{\js \in J\mid \frac{1}{\eta_1}c_{\min} < d_{\js} \le \frac{1}{\eta_2}c_{\min}\}$, for some $\eta_1 > \eta_2 \ge 1$. Then there exists a valid \rtree\ packing of $\Hs$ using less than $\frac{4\eta_1(\eta_2 +1)}{\eta^{2}_2}r + 4$ tree-rounds.
\end{lemma}
\begin{proof}
We scale up the demand of each job in $\Hs$ to $\frac{1}{\eta_2}c_{\min}$ and scale down the capacity of each edge in $E$ to the nearest integral multiple of $\frac{1}{\eta_2}c_{\min}$. Let $\Hs'$ denote the new set of identical demand (of $\frac{1}{\eta_2}c_{\min}$) jobs and let $I' = (c'_{e_1},\ldots,c'_{e_m})$ denote the new profile. For $e\in E$, let $l^{\Hs}_e$ and $r^{\Hs}_e$ denote the original load and congestion of edge $e$ due to the jobs in $\Hs$ and let $l^{\Hs'}_e$ and $r^{\Hs'}_e$ denote the new load and congestion, respectively. Then $l^{\Hs'}_e < \frac{\eta_1}{\eta_2}l^{\Hs}_e$ and $c'_e \ge \frac{\eta_2}{\eta_2 + 1}c_e$. Thus $r^{\Hs'}_e = \left \lceil \frac{l^{\Hs'}_e}{c'_e} \right \rceil \le \left \lceil \frac{\frac{\eta_1}{\eta_2}l^{\Hs}_e}{\frac{\eta_2}{\eta_2 + 1}c_e} \right \rceil < \frac{\eta_1 (\eta_2 + 1)}{\eta^2_2}\frac{l^{\Hs}}{c_e}+1 \le \frac{\eta_1 (\eta_2 + 1)}{\eta^2_2}r^{\Hs}_e + 1$.

Finally, we scale all demands and edge capacities by $\frac{1}{\eta_2}c_{\min}$ so that now all jobs have unit demands and all capacities are integers. Applying \Cref{theorem:Chekuristheorem} to this set of jobs, we obtain the number of tree-rounds used is at most $4r^{\Hs'} < \frac{4\eta_1(\eta_2 +1)}{\eta^{2}_2}r + 4$.
\end{proof}

Let $\Js_L := \{\js \in J\mid d_{\js} > \mathfrak{b}_{\js}/5\}$ and $\Js_s := J\setminus \Js_L$.

\begin{lemma}
\label{lemma:31approx}
The jobs in $\Js_L$ can be packed into at most $31r+6$ \textit{tree-rounds}.
\end{lemma}
\begin{proof}
Let $\Qs_M := \{\js \in \Js_L \mid \frac{1}{5}c_{\min} < d_{\js} \le \frac{1}{2}c_{\min}\}$ and $\Qs_L := \Js_L \setminus \Qs_M$. Apply \Cref{lemma:treelemma} to $\Qs_M$ with $\eta_1 = 5$ and $\eta_2 = 2$ and to $\Qs_L$ with $\eta_1 = 2$ and $\eta_2 = 1$. This yields valid packings of $\Qs_M$ and $\Qs_L$ using at most $15r+3$ and $16r+3$ tree-rounds, respectively (since the number of rounds is an integer and \Cref{lemma:treelemma} is strict). Thus the total number of tree-rounds used $\le 31r + 6$.
\end{proof}

\begin{lemma}
\label{lemma:18approx}
The jobs in $\Js_S$ can be packed into at most $18r$ \textit{tree-rounds}.
\end{lemma}
\begin{proof}
First, we introduce some terminologies. An edge $e\in E(G_{\text{tree}})$ is said to be of \textit{class} $k$ if $(5/2)^k \le c_e < (5/2)^{(k+1)}$ and we denote $\text{cl}(e):=k$. For any $u,v \in V(G_{\text{tree}})$, the \textit{critical edge} of $P_{u-v}$, denoted by $\text{crit}(P_{u-v})$ is defined as the first edge having the minimum class among all the edges of $P_{u-v}$. 

We again fix any vertex ($v_{\text{root}}$) as the root of $G_{\text{tree}}$.
As previously, for any $\js \in J$, let $\theta_{\js}$ denote the least common ancestor of $v_{s_{\js}}$ and $v_{t_{\js}}$.
 We maintain $18r$ tree-rounds, $\Rs^1, \Rs^2,\ldots,\Rs^{18r}$. For any job $\js \in \Js_S$, let $e_{\js}^1 := \text{crit}(P_{\theta_{\js}-v_{s_{\js}}})$ and $e_{\js}^2 := \text{crit}(P_{\theta_{\js}-v_{t_{\js}}})$. We consider the jobs in $\Js_S$ in non-decreasing order of the levels of $\theta_{\js}$ and place $\js$ in a tree-round $\Rs$ in which both the following conditions hold:
\begin{enumerate}[(i)]
\item The sum of demands of all jobs that have already been assigned to $\Rs$ and pass through $e_{\js}^{1}$ is at most $c_{e_{\js}^{1}}/9$.
\item The sum of demands of all jobs that have already been assigned to $\Rs$ and pass through $e_{\js}^{2}$ is at most $c_{e_{\js}^{2}}/9$. 
\end{enumerate}

First, we show that such a tree-round must always exist. Suppose to the contrary, there exists a job $\js$ that cannot be assigned by the above algorithm. Then in each of the $18r$ tree-rounds, either (i) or (ii) must fail. For any tree-round $\Rs$ and edge $e\in E(G_{\text{tree}})$, let $l_{e}^{\Rs}$ denote the sum of demands of jobs assigned to $\Rs$ that pass through $e$. Let $\chi$ be the number of tree-rounds in which $l_{e_{\js}^1}^{\Rs^i} > \frac{1}{9}c_{e_{\js}^1}$. Thus in at least $18r-\chi$ tree-rounds, $l_{e_{\js}^2}^{\Rs^i} > \frac{1}{9}c_{e_{\js}^2}$. So $l_{e_{\js}^1} \ge \sum \limits_{i=1}^{18r}l_{e_{\js}^1}^{\Rs^i} > \chi\cdot \frac{1}{9}c_{e_{\js}^1}$. Since $r \ge l_{e_{\js}^1}/c_{e_{\js}^1}$, we get $\chi < 9r$. Again $l_{e_{\js}^2} \ge \sum \limits_{i=1}^{18r}l_{e_{\js}^2}^{\Rs^i} > (18r-\chi)\cdot \frac{1}{9}c_{e_{\js}^2}$ and since $r \ge l_{e_{\js}^2} /c_{e_{\js}^2}$, we get $\chi > 9r$, a contradiction.

Now we show that the above algorithm produces a valid packing. Consider any tree-round $\Rs$ and edge $e = \{u,v\}$, such that $\text{level}(v)=\text{level}(u)+1$. Let $\Ss_e^{\Rs}$ denote the set of jobs assigned to $\Rs$ that pass through $e$. Let $\js_{\text{up}} \in \Ss_e^{\Rs}$ be the last job such that at least one of $e_{j_{\text{up}}}^{1}$ or $e_{j_{\text{up}}}^{2}$ lies on $P_{v_{\text{root}}-v}$. Thus the sum of demands of all such jobs such that at least one of $e_{j_{\text{up}}}^{1}$ or $e_{j_{\text{up}}}^{2}$ lies on $P_{v_{\text{root}}-v}$, including $\js_{\text{up}}$, is at most \[\frac{(5/2)^{\text{cl}(e)+1}}{9} + \frac{c_e}{5} \le \frac{43}{90}c_e.\]
       
All the remaining jobs must satisfy $\text{level}(\theta_{\js}) > \text{level}(u)$. 
So the critical edge is of lower level. 
Thus the sum of demands of all such jobs is at most \[\sum \limits_{i=0}^{\text{cl}(e)-1} \Big (\frac{(5/2)^{i+1}}{9} + \frac{(5/2)^{i+1}}{5}\Big ) \le \frac{14}{45}\sum \limits_{i=0}^{\text{cl}(e)-1} (5/2)^{i+1} \le \frac{14}{27}c_e.\]
    
Since ${43}/{90} + {14}/{27} < 1$, the total sum of demands of jobs in $\Ss_e^{\Rs}$ does not exceed $c_e$. Hence, the packing is valid.
\end{proof}

Note that \Cref{lemma:31approx} yields a packing of $\Js_L$ using at most $31r+6 \le 37r$ tree-rounds. Together with \Cref{lemma:18approx}, we have the following theorem.

\begin{theorem}
There exists a polynomial time asymptotic (resp. absolute) 49- (resp. 55-) approximation algorithm for \rtree.
\end{theorem}


We note that with some careful choice of parameters, the asymptotic approximation ratio could be improved to 48.292, but it is unlikely to be improved further using the present technique.


\end{document}